\documentclass[a4paper, reqno]{amsart}
\usepackage[noabbrev]{cleveref}
\usepackage{appendix}
\crefname{appsec}{Appendix}{Appendices}
\usepackage{amssymb, amsfonts, amsthm, amsmath, mathtext, latexsym,
graphicx, physics}
\usepackage{color}

\numberwithin{equation}{section}
\newcommand{\R}{{\mathbb R}}

\newcommand{\Z}{{\mathbb Z}}

\newcommand{\C}{{\mathbb C}}
\renewcommand{\mod}{{\rm mod}\,}
\newcommand{\re}{{\rm Re}\,}
\newcommand{\im}{{\rm Im}\,}

\newcommand{\res}{{\rm res}\, }

\theoremstyle{plain}
\newtheorem{theorem}{Theorem}[section]
\newtheorem{lemma}{Lemma}[section]
\newtheorem{proposition}{Proposition}[section]
\newtheorem{corollary}{Corollary}[section]

\theoremstyle{definition}

\newtheorem{definition}{Definition}[section]
\title{Difference equations in the complex plane : 
quasiclassical asymptotics and Berry phase} 
\author{A. Fedotov}
\address{Saint Petersburg State University,
7/9, Universitetskaya nab., 
Saint Petersburg, 199034, Russia}
\email{a.fedotov@spbu.ru}
\author{E. Shchetka}
\address{Saint Petersburg Department of V.A.Steklov Institute of Mathematics of 
the Russian Academy of Sciences,  27, Fontanka, Saint Petersburg, 191023, Russia;
Chebyshev Laboratory, 14th Line 29B, Vasilyevsky Island, 
Saint Petersburg, 199178, Russia.}
\email{e.shchetka@spbu.ru}
\keywords{Difference equation, complex WKB method, 
quasiclassical asymptotics, geometric phase}
\thanks{The work of A.F. was supported by Russian foundation 
for basic research under the grant 17-01-00668.
The work of E.S. was supported by ``Native towns'', 
a social investment program of PJSC ``Gazprom Neft''}
\begin{document}
\begin{abstract}
We study solutions to the difference equation
$\Psi(z+h)=M(z)\Psi(z)$ where  $z$ is a complex variable, 
$h>0$ is a  parameter, and  $M:\mathbb{C}\mapsto SL(2,\mathbb{C})$ 
is a given analytic  function. We describe the asymptotics of  
its analytic solutions as $h\to 0$. The asymptotic formulas contain an analog of 
the geometric (Berry) phase well-known in the quasiclassical analysis 
of differential equations. 
\end{abstract}
\maketitle
%
%
%
\section{Introduction}\label{ch:intro}
\noindent
For $M:\mathbb{C}\mapsto SL(2,\mathbb{C})$ being a given analytic
function, we consider  the equation
\begin{equation}\label{main}
\Psi(z+h)=M(z)\Psi(z)
\end{equation}
where $z$ is a complex variable, and $h>0$ is a parameter. 
We describe asymptotics of analytic vector
solutions $\Psi$ to~\eqref{main} as $h\to 0$.
\\
Formally, $\Psi(z+h)=e^{h\frac{d}{dz}}\psi(z)$, and being a small parameter 
in front of the derivative, $h$ can be regarded as a quasiclassical asymptotic
parameter. 
\\
The quasiclassical asymptotics of solutions to the ordinary differential equation
\begin{equation}\label{differential-eq}
  ih\frac{d\Psi}{dx}(x)=M(x)\Psi(x)
\end{equation}
as $h\to 0$ are described by means of the famous WKB (Wentzel, 
Kramers and  Brillouin) method. There is a huge literature devoted 
to  this method and its applications. If $M$ is analytic,  one uses a  method often called  
the  complex WKB method, see, e.g., chapters 3 and 5 in~\cite{Fe:93}
and  chapter 7 in~\cite{W:87}. This method allows to study  
solutions to~\eqref{differential-eq} on the complex plane. Even when the input 
problem does not require to go into the complex plane, one uses this 
method  to simplify the analysis: it allows to go around, say, 
turning points or singularities of solutions located on the real line, and
to compute the asymptotics of  their Wronskians in the domains where they 
are easy to be  computed.  The latter makes the  complex WKB method 
very efficient for computing exponentially small quantities.  Actually, 
in~\cite{Fe:93} one can find various  interesting examples of problems solved 
using this method.
\\
To study difference equations on the real axis in the quasiclassical approximation, 
one uses methods similar to the classical WKB methods (e.g.,~\cite{G-at-al}), 
pseudodifferential operator theory (e.g., \cite{H-S:88}), and Maslov's canonical 
operator method (e.g. \cite{Dobro}).
\\
For difference equations on the complex plane, a complex WKB method can play 
the same role as for differential ones, and an analog of the complex 
WKB method  for difference equations is being  developed in~\cite{B-F:94, 
F-Shch:17, F-Shch:18, F-K:18, F-K:18a} and in the present paper. In~\cite{B-F:94, 
F-Shch:17, F-Shch:18, F-K:18, F-K:18a} the authors develop an analog of the complex 
WKB method for the one-dimensional difference Schr\"odinger equation
\begin{equation}\label{scalar-eq}
\psi(z+h)+\psi(z-h)+v(z)\psi(z)=0,
\end{equation}
where $v$ is an analytic function.
In this paper we extend this  method to
the matrix difference \cref{main},  get asymptotic formulas 
for its solutions, and, in particular, find an analog of the geometric phase (Berry phase)
well-known in the case of differential equations, see~\cite{B:84, S:83}.
\\
Our work is motivated by the analysis of the spectrum of the Harper operator 
acting in $L_2({\mathbb R})$ by the formula $H\psi(z)=\psi(z+h)+\psi(z-h)+
2\lambda \cos z\,\psi(z)$.  This operator arises in the solid state physics when 
studying  an electron in a crystal submitted to a magnetic field, see, e.g., 
the introduction sections in~\cite{Wi, GHT:89} and references therein. For irrational 
$h$ the spectrum coincides with one of the famous almost Mathieu operator, 
and is a Cantor set, see, e.g.~\cite{A-J:09}. In~\cite{Wi} heuristically, and 
in~\cite{H-S:88} rigorously,  the authors obtain  in the quasiclassical 
approximation a description of the spectrum similar to one of the classical 
Cantor set:  they discovered step by step sequences of smaller and smaller 
spectral gaps. Note that the gaps of each sequence appear to be exponentially 
small with respect to the gaps of the previous one. To study the geometrical 
properties of the spectrum, Buslaev and Fedotov have suggested a renormalization 
approach based on ideas of the Floquet theory, see~\cite{F:13}. A crucial role in 
their analysis  is played by the minimal entire solutions to the Harper equation 
$\psi(z+h)+\psi(z-h)+2\lambda \cos z\,\psi(z)=E\psi(z)$, where $E$ is a spectral 
parameter. To study them in the quasiclassical approximation,  a version of the 
complex WKB for difference equations was developed, see~\cite{B-F:94, B-F:94a, F:13}.
The analysis of geometrical properties of the spectrum  also requires to analyze 
solutions to matrix difference equations of the form~\eqref{main} with complex 
coefficients (see, section 2.3.2 in \cite{F:13}). Similar problems arise when, instead 
of the Harper operator, one studies more general difference and differential  
one-dimensional quasiperiodic Schr\"odinger equations with two frequencies, 
see~\cite{F:13}.
\\
Of course, difference equations in $\mathbb C$  with small $h$ arise  in many 
other fields of mathematics and physics. For example, they appear  in  the study 
of diffraction of classical waves by wedges, see, e.g., equation (2.1.2)  
in~\cite{L-Z:03}.  A small shift parameter arises in the case of  narrow wedges 
(as the shift parameters appearing in these problems are proportional to the angles 
of the wedges, see~\cite{L-Z:03}). 
\\
In this paper we study two cases: the case when $M$ is analytic in a bounded 
domain, and the case when $M$ is a  trigonometric polynomial. 
\\
In the next section we describe the main objects of the complex WKB method for \cref{main}:
the complex momentum, geometric phase, canonical curves and canonical domains.
Then,  we formulate and discuss two our theorems  on the existence of analytic solutions 
to~\eqref{main} having a simple quasiclassical behavior  in certain complex domains. 
In \cref{s:gph} we turn to the geometric phase appearing in the asymptotic 
formulas. It is very natural to consider it as an 
integral of a meromorphic differential (meromorphic differential 1-form) on a Riemann 
surface, and we study this differential in details. In \cref{s:proof1} we prove the 
existence of analytic solutions having simple  asymptotic behavior in bounded 
domains, and in \cref{s:proof2} we turn to the case where $M$ is a trigonometric 
polynomial. 
\\
Our results  were announced in a short note~\cite{F-Shch:17DD}  (conference proceedings).
\section{The main construction of the complex WKB method}
\label{s:basics}
We begin with formulating our assumptions on the matrix $M$.
\subsection{Our assumptions}\label{ss:assumptions}
We assume that either the domain of analyticity of $M$ is bounded or   
$M$ is a trigonometric polynomial, i.e., 
\begin{equation}\label{eq:M:poly}
M(z)=\sum_{j=-k}^l M_j e^{2\pi i j z},\quad z\in \C,
\end{equation}
\noindent
where $M_j$ are Fourier coefficients.
We do not consider the degenerate case where 
$M_{12} M_{21}\equiv0$  (in this case~\cref{main} can be solved explicitly). 
\\
In the case of~\eqref{eq:M:poly}, we assume also that
$k,l>0$, that   $ \tr M_{-k}\tr M_l \ne 0$, and that
\begin{equation}
  \label{eq:d/a}
  M_{22}(z)/M_{11}(z) \text{ stays bounded as } |\im z|\to\infty.
\end{equation}
The last hypothesis can be removed and is made just for the sake of simplicity.  
\subsection{The complex momentum}
The {\it complex momentum} $p$ is the multivalued analytic 
function defined in the domain of analyticity of $M$ by the formula 
\begin{equation}\label{eq:p}
2\cos p(z)=\tr M(z), \quad z\in D.
\end{equation} 
The branch points of $p$ satisfy the equations $\tr M(z)=\pm 2$. 
We call points where $\tr M(z)\in\{\pm 2\}$ turning points.
We say that a subset of the domain of analyticity of $M$   {\it regular} if it  contains no turning points.
\\
As $\det M(z)\equiv 1$, the eigenvalues of $M(z)$ are equal to $e^{\pm ip(z)}$.
If $z$ is regular, one has $e^{ip(z)}\ne e^{-ip(z)}$.
\subsection{The geometric phase}\label{sss:geom-phase}
Let $R\subset \C$ be a regular simply connected domain.  We fix in $R$ an analytic branch $p$ 
of the complex momentum.
\\
Let $r^\pm:R\mapsto \mathbb C^2$ be two nontrivial analytic functions
satisfying the equations
\begin{equation}
  \label{d:rev}
  M(z) r^\pm(z)=e^{\pm ip(z)}r^\pm(z),\quad z\in R.
\end{equation}
We set
\begin{equation}\label{eq:left-right}
l^\pm(z)=(r^\mp)^{\mathrm T}(z)\, \sigma,\qquad 
\sigma=\begin{pmatrix} 0 & -1 \\ 1 & 0 \end{pmatrix},
\end{equation}
where ${\;\cdot\;}^{\mathrm T}$ denotes transposition. 
By~\Cref{le:left-right}, we have
\begin{equation}
  \label{d:lev}
  l^\pm(z) M(z)=e^{\pm ip(z)}l^\pm(z),\quad z\in R.
\end{equation}
If $r^+(z)\ne 0$ ( $r^-(z)\ne 0$ ), then $r^+(z)$ ( resp., $r^-(z)$ ) is  a 
right eigenvector of $M(z)$, and  $l^-(z)$ ( resp., 
$l^+(z)$ ) is its left  eigenvector.
\\
The analytic functions $z\mapsto l^\pm(z) r^\pm(z) $ are 
not identically zero (in view of \eqref{eq:det:1}), and 
 we define in $R$ two meromorphic differentials  $\Omega_\pm$ by the formulas
\begin{equation}\label{def:omega}
\Omega_\pm(z)=\mp\frac{i}{2}d p(z) -\frac{l^\pm(z)\,d\,r^\pm(z)}
{l^\pm(z) r^\pm(z)}.
\end{equation}
Let us note that the poles of $\Omega_\pm$ are located at points where $r_\pm(z)=0$ 
(in view of \Cref{th:V}). Let $z_0\in R$, and $r^\pm (z_0)\ne 0$. 
The integrals $\int_{z_0}^z\Omega_\pm$ are called {\it  geometric phases}.
Very close objects are well-known in the  WKB analysis of differential equations, 
see~\cref{differential_eq}. But, it looks like their properties has not been 
systematically studied as properties of functions of the complex variable.  
\\
We study  $\Omega_\pm$ in~\cref{s:gph}.
For two column vectors $u,v\in\mathbb{C}^2$, we denote by  $(u\;v)$   
the 2$\times$2-matrix with the columns $u$ and $v$.
In \cref{s:proof:thm:Omegas} we check

\begin{theorem}\label{th:V} Let $z_0\in R$ and  $r^\pm(z_0)\ne 0$.
In the domain $R$ each of the functions
\begin{equation}
  \label{d:ev}
 V^\pm: z\mapsto \exp\left(\int_{z_0}^z\Omega_\pm\right)\,r^\pm(z)
\end{equation}
is analytic, does not vanish and is independent  of the choice of $r^\pm$
up to a constant factor. Moreover, one has
\begin{equation}
  \label{eq:det}
  \det(V^+(z)\;V^-(z))=\det(r^+(z_0)\;r^-(z_0))\ne 0.
\end{equation}
\end{theorem}
\noindent 
We call $V^\pm$ analytic eigenvectors of $M$ normalized at $z_0$.
\\
The facts that $V^\pm$  are independent of the choice of $r^\pm$ and 
satisfy~\eqref{eq:det}, are proved by means of the ideas used to check 
similar facts in the case of differential equations on $\R$, see, e.g., 
section 3 of chapter 5 in~\cite{Fe:93}. 
\subsection{The canonical curves}\label{sss:CC}
For $z\in\mathbb C$, we let $x=\mathrm{Re}\,z$, $y=\mathrm{Im}\,z$.
\\
A curve $\gamma\subset \mathbb C$ is called \textit{vertical} if, along $\gamma$, 
$x$ is a piecewise continuously differentiable function of $y$.
We say that $\gamma$ is infinite if along it $y$ increases from $-\infty$ to $\infty$.
\\
Let $R$ be a regular simply connected domain, and $z_0\in R$. We fix in $R$ an 
analytic branch of the complex momentum $p$.  Let $\gamma\subset R$ be a vertical 
curve, and let $z(y)$ be the point of $\gamma$ with the imaginary part equal to $y$.
This curve is called {\it canonical} with respect 
to the branch $p$ if, at all the points of $\gamma$ where 
$z'$ exists, one has
\begin{equation}\label{def:can}
\frac{d}{dy} \, \mathrm{Im}\,\int_{z_0}^{z(y)} p(z)\,dz > 0,
\quad\text{and}\quad 
\frac{d}{dy}\, \mathrm{Im}\,\int_{z_0}^{z(y)} (p(z)-\pi)\,dz< 0,
\end{equation}
and, at the points where $dz/dy$ is discontinuous,  these inequalities hold
for  the left and  right derivatives.
\subsection{The canonical domains}\label{sss:CD}
The definitions of the bounded and unbounded canonical domains are 
slightly different.
\subsubsection{Bounded canonical domains}
We call a domain horizontally connected if, for any its two points
having one and the same imaginary part, the straight line segment 
that connects them is contained in this domain.
\\
Let $K$ be a bounded regular horizontally connected domain,  
$p$ be a branch of the complex momentum analytic in $K$,    
and $z_1,\,z_2$ be two regular points of the boundary of $ K$.  
We call $K$ \textit{canonical} with respect to  $p$ if,  $\forall z\in K$, 
there is a curve $\gamma$  connecting  $z_1$ and $z_2$ in $K$, containing $z$
and canonical with respect to $p$. 
\subsubsection{Unbounded canonical domains}
If $M$ is a trigonometric polynomial, we consider the unbounded
canonical domains that contain infinite vertical curves.
\\
We call a domain horizontally bounded if $|\re z|$ stays bounded for all $z$ in it.
\\
Let $K\subset\C$ be an unbounded  regular, horizontally connected and horizontally 
bounded domain, let $p$ be a branch of the complex momentum analytic in it.
We call  the domain $K$ {\it canonical} with respect to  $p$ if for any 
$z\in K$ there is an infinite curve $\gamma\subset K$ canonical with respect to $p$  
and containing $z$. 
\subsection{Main theorems}
Below $K\subset \mathbb C$ is a  domain canonical with respect to a branch  $p$, 
and  $V^\pm$ are analytic eigenvectors of $M$ normalized at
$z_0\in K$ and corresponding to the eigenvalues $e^{\pm i p(z)}$.
\subsubsection{Locally uniform asymptotics} 
Let us recall that an asymptotic representation is locally uniform in a 
domain $D$ if it is uniform in any fixed compact subset of $D$.
\\
First, we describe locally uniform asymptotics of solutions to~\eqref{main}.
One has
\begin{theorem}\label{th:main:local} 
For sufficiently small $h$, in $K$ there exist $\Psi^\pm$,
two analytic solutions to~\eqref{main}, admitting  the  following locally 
uniform asymptotic representations : 
\begin{equation}\label{as:main}
\Psi^\pm(z)= e^{\pm\frac{i}{h}\int_{z_0}^zp(z)\,dz}
\left(V^\pm(z)+O(h)\right),
\quad h\to 0.
\end{equation}
\end{theorem}
\subsubsection{Asymptotics in unbounded domains}
Here, we concentrate on the case where $K$ is an unbounded canonical domain, and
describe the behavior of the  solutions $\Psi^\pm(z)$ from 
\Cref{th:main:local}  for large $|\im z|$. 
\\
For a fixed $\delta>0$, we call the domain $K$ without the $\delta$-neighborhood 
of its boundary an admissible subdomain of $K$.
\begin{theorem}\label{th:main:global} Let, in the case of the previous theorem,
the domain $K$ be unbounded, and $A$ be its admissible subdomain.
For sufficiently large $Y$, for $|\im z |\ge Y$ the solutions $\Psi^\pm$ admit in $A$
the following uniform asymptotic representations :
\begin{equation}\label{as:main:global}
\Psi^\pm(z)= e^{\pm\frac{i}{h}\int_{z_0}^zp(z)\,dz+g}
\begin{pmatrix}V_1^\pm(z)\,(1+O(h))\\ 
V_2^\pm(z)\,(1+O(h))\end{pmatrix},
\quad h\to 0.
\end{equation}
Here $|g|\le C\,(1+|z|)\,h$ with a constant $C>0$ independent of $h$.
If $\frac{M_{22}(z)}{M_{11}(z)}\to 0$ as $|\im z|\to\infty$, then $|g|\le C\,h$.
\end{theorem}
\noindent In \eqref{as:main:global} the $O(h)$ decay exponentially  as $|y|\to\infty$
(for $\Psi^+$ see \eqref{eq:PsiX} and~\eqref{est:V}).
\subsubsection{Known results for  differential equations}\label{differential_eq}
For \cref{differential-eq}, for sufficiently small $h$ one  constructs vector solutions $\Psi_j$,
\ $j=1,2$, such that
\begin{equation}\label{as:dif-eq}
\Psi_{j}(x)\sim e^{\frac{i}{h}\int \limits_{x_0}^x p_j \,dx-
\int\limits_{x_0}^x \frac{l_j\; r_j'}{l_jr_j}\,dx}\,r_j, \quad h\to 0.
\end{equation}
where  $p_j$ are eigenvalues of $M$, and $l_j$ and $r_j$ 
are the corresponding left and right eigenvectors,  see section 4 of chapter 5
in \cite{Fe:93}, and we have written  only the leading terms of the asymptotics.\\
The expressions $-\int_{x_0}^x \frac{l_j\; r_j'}{l_jr_j}\,dx$, \  
$j=1,2$ are often called {\it geometric phases} or {\it Berry phases}  (see \cite{B:84}) 
and have a well-known geometric interpretation (see \cite{S:83}). 
\subsubsection{An example}
Let us consider  the scalar difference \cref{scalar-eq} with an analytic 
function $v$. A vector function $\Psi$ satisfies~\eqref{main} with the matrix
$
M(z)=\begin{pmatrix}
-v(z) & -1
\\
1 & 0 
\end{pmatrix}
$
if and only if $\Psi(z)=\begin{pmatrix}\psi(z) \\ \psi(z-h)\end{pmatrix}$ where
$\psi$ is a solution to~\eqref{scalar-eq}. 
\\
Let us  deduce the quasiclassical 
asymptotics of solutions to~\eqref{scalar-eq} from Theorem~\ref{th:main:local}.
\\
For the above $M(z)$,  the complex momentum is defined 
by the relation $2\cos p(z)+v(z)=0$, and  as eigenvectors of $M(z)$ one can choose 
\begin{equation*}
r^\pm=\begin{pmatrix}1\\e^{\mp ip(z)} \end{pmatrix}\quad
\text{and} \quad
l^\pm=\begin{pmatrix}e^{\pm ip(z)}&-1 \end{pmatrix}.
\end{equation*}
Then
\begin{multline*}
\int_{z_0}^z\Omega_\pm=\mp\frac{i}{2}\int_{z_0}^z\left(p'(s)+ 
\frac{2p'(s)e^{\mp ip(s)}}{e^{\pm ip(s)}-e^{\mp ip(s)}}\right) \,ds
\\
=-\frac{i}2\int_{z_0}^z 
\frac{p'(s)\left(e^{ip(s)}+e^{-ip(s)}\right)}{e^{ip(s)}-e^{-ip(s)}} \,ds
=-\left.\ln \sqrt{\sin p}\, \right|_{z_0}^{z}.
\end{multline*}
This leads to the following  formulas for two analytic solutions 
to \eqref{scalar-eq} :
\begin{equation}\label{scalar-eq:as}
  \psi^\pm(z)=\frac1{\sqrt{\sin p(z)}}e^{\pm\frac{i}h\int_{z_0}^zp(z)\,dz+O(h)}.
\end{equation}
These formulas were obtained in, e.g.,~\cite{F-Shch:18}.
\section{The meromorphic differentials $\Omega_\pm$}\label{s:gph}
\subsection{Preliminaries}
\subsubsection{The left and right eigenvectors of 
unimodular $2\times 2$-matrices}\label{ss:ev}
Here, we assume only that  $M\in SL(2,{\mathbb C})$. Then the eigenvalues of $M$ are of 
the form $e^{\pm i p}$, where $p$ is a complex number. We assume that  $e^{ip}\ne e^{-ip}$.
Let  $r^\pm$ be right eigenvectors of $M$ corresponding to the eigenvalues 
$e^{\pm i p}$. One has
\begin{lemma}\label{le:left-right} The row vectors defined by formula~\eqref{eq:left-right}
are left eigenvectors corresponding to the eigenvalues   $e^{\pm ip}.$
\end{lemma}
\begin{proof}
As $M\in SL(2,\mathbb{C})$, we have
\begin{equation*}
\sigma\,M=\left(M^{-1}\right)^{\mathrm T}\sigma.
\end{equation*}
Hence
\begin{equation*}
l^\mp \,M= {r^\pm}^{\mathrm T}\sigma M
={r^\pm}^{\mathrm T} \left(M^{-1}\right)^{\mathrm T}\sigma 
=\left(M^{-1} r^\pm\right)^{\mathrm T}\sigma
=e^{\mp ip}{r^\pm}^{\mathrm T}\sigma=e^{\mp ip} \,l^\mp.
\end{equation*}
\end{proof}
\noindent
Below, $l^\pm$ are always defined by \eqref{eq:left-right}.
\Cref{le:left-right} implies 
\begin{lemma}\label{le:det}
One has
\begin{equation}\label{eq:det:1}
 l^\pm r^\pm\,= \pm \det (r^+\;r^-),\quad \text{and}\quad l^\pm r^\mp=0,
\end{equation}
where $(r^+\;r^-)$ is the matrix with the columns $r^+$ and $r^-$.
\end{lemma}
\begin{proof} One has
$$l^\pm r^\pm={r^\mp}^{\mathrm T}\sigma r^\pm
=-(\sigma r^\mp)^{\mathrm T}r^\pm
=-\left(\sigma r^\mp,\,r^\pm\right)_{\mathbb{R}^2}
=-\det(r^\mp\;r^\pm)=\pm \det (r^+\;r^-).
$$
This proves the first two equalities. The remaining  two
are proved similarly.
\end{proof}
\noindent \Cref{le:det} can be equivalently formulated in the following form.
Let us denote by $\begin{pmatrix}l^+\\ -l^-\end{pmatrix}$ the matrix 
with the rows $l^+$ and $-l^-$. One has
\begin{corollary}\label{cor:det}
\begin{equation}\label{eq:det:matrix}
  \begin{pmatrix}l^+\\ -l^-\end{pmatrix}\,(r^+\;r^-)=\det (r^+\;r^-)\cdot I.
\end{equation}
\end{corollary}
\subsubsection{Analytic solutions to~\cref{d:rev} and 
differentials $\Omega_\pm$}\label{ss:analytic-ev}
Let us come back to~\eqref{main}.
Let $R\subset D$ be a simply connected regular domain, 
and let $p$ be a branch of the complex momentum analytic in $R$.
\\
By \Cref{th:V}, up to constant factors,   the vectors $V_\pm(z)$ are independent 
of the choice of $r_\pm$, analytic solutions  to~\eqref{d:rev}, used to construct 
them.  Throughout this paper  $r^\pm$ are vectors given by the formulas:
\begin{equation}\label{d:r-pm}
   r^\pm(z)=\begin{pmatrix} M_{12}(z) \\ e^{\pm ip(z)}-M_{11}(z) \end{pmatrix}.
\end{equation}  
\noindent One has
\begin{equation}
  \label{f:det}
  \det( r^+(z)\ \  r^-(z))=-2iM_{12}(z)\sin p(z).
\end{equation}
As $R$ is regular, $\sin p(z)\ne 0$. So, the determinant 
vanishes only at zeros of $M_{12}$.
\\ 
Actually, one has
\begin{lemma}\label{le:b:zeros}
If $M_{12}(z)=0$ at $z\in R$, then one and only one of the vectors $ r^\pm(z)$ equals zero.   
\end{lemma}
\begin{proof}
As $M_{12}(z)=0$, the numbers $M_{11}(z)$ and $M_{22}(z)$ are eigenvalues of $M(z)$.
As $z$ is regular, $M_{11}(z)\ne M_{22}(z)$.
So, either $M_{11}(z)=e^{ip(z)}$ and $M_{11}(z)\ne e^{-ip(z)}$ or
 $M_{11}(z)=e^{-ip(z)}$ and $M_{11}(z)\ne e^{ip(z)}$. This and~\eqref{d:r-pm} imply
the statement.
\end{proof}
\noindent 
Now, we define two row vectors $ l^\pm (z)$  by the formula~\eqref{eq:left-right}. One has
\begin{equation}\label{d:l-pm}
   l^\pm (z)=\left( e^{\mp i p(z)}-M_{11}(z)\ \  -M_{12}(z)\right).
\end{equation}
Let us compute the differentials $\Omega_\pm$ corresponding 
to the chosen $r^\pm$.
\\ 
For our choice of $ r^\pm$, formulas~\eqref{cor:det} and~(\ref{f:det}) imply that
\begin{equation}\label{f:denominator}
 l^\pm(z) r^\pm(z)=\mp 2iM_{12}(z)\sin p(z).
\end{equation}
Using~\eqref{f:denominator},~\eqref{d:r-pm},~\eqref{d:l-pm}, and the definitions 
of $ \Omega_\pm$, see~(\ref{def:omega}), we prove that
\begin{equation}\label{le:f:Omega}
   \Omega_\pm(z)=\mp\frac{idp(z)}2\;
\pm\frac{\left(e^{\mp ip(z)}-M_{11}(z)\right)\, d\,\ln M_{12}(z)- 
d\,\left(e^{\pm ip(z)}-M_{11}(z)\right)}{2i\sin p(z)}.
\end{equation}
In the rest of this paper, $\Omega_\pm$ are the differentials given by these formulas.
\subsection{Differentials $ \Omega_\pm$ in regular domains}
Let, again, $R$ be  a regular domain. 
\\
By~\eqref{le:f:Omega},  \ $\Omega_\pm$ can have poles in $R$ only at 
the points where $M_{12}(z)=0$ (as $e^{\pm i p(z)}$ differ in $R$). 
Let $z_0\in R$, and  let $M_{12}(z_0)\ne 0$. One has
\begin{proposition} \label{pro:omega:poles}
Pick $s\in\{\pm\}$.  Let $z_*\in R$ and $M_{12}(z_*)=0$. 
If $r^s(z_*)\ne 0$, then $ \Omega_s$ is holomorphic at $z_*$.
If $r^s(z_*)=0$, then the function $z\mapsto 
\exp\left(\int_{z_0}^z \Omega_s(z)\right)\, r^s(z)$
is analytic and does not vanish at $z_*$. In this case,
in a neighborhood of $z_*$, one has 
\begin{equation}\label{eq:omega:near-a-zero-of-b}
   \Omega_s(z)= -d\ln M_{12}(z)+ \text{a holomorphic differential}.
\end{equation}
\end{proposition}
\noindent
This proposition and \Cref{le:b:zeros} give quite a complete
description of the poles of $\Omega_\pm$ in a regular domain.
\begin{proof}For the sake of definiteness, we assume that $s=+$. 
By~\eqref{le:f:Omega}, one has   
\begin{equation}\label{eq:poles:1}\textstyle
  \Omega_+(z)= \frac{e^{-ip(z)}-M_{11}(z)}{2i\sin p(z)}d\ln M_{12}(z)+
\text{(a holomorphic differential)}, \quad z\sim z_*.
\end{equation}
First, let us assume that $r^+(z_*)\ne 0$. Then, by \Cref{le:b:zeros}, 
$r^-(z_*)=0$, and, therefore, by~(\ref{d:r-pm})  one has 
$M_{11}(z_*)=e^{-ip(z_*)}\ne e^{ip(z_*)}$. 
But
\begin{equation}\label{eq:poles:2}
  \frac{e^{-ip(z)}-M_{11}(z)}{M_{12}(z)}=-\frac{M_{21}(z)}{e^{ip(z)}-M_{11}(z)}.
\end{equation}
Indeed, as $2\cos p(z)=M_{11}(z)+M_{22}(z)$, and 
$M_{11}(z)M_{22}(z)-M_{12}(z)M_{21}(z)=1$, 
\begin{equation*}
  (e^{-ip(z)}-M_{11}(z))(e^{ip(z)}-M_{11}(z))=
1-2\cos p(z) M_{11}(z)+M_{11}(z)^2=-M_{12}(z)M_{21}(z).
\end{equation*}
Formulas~\eqref{eq:poles:1} and~\eqref{eq:poles:2} imply that $ \Omega_+$ 
is holomorphic  at $z_*$. 
\\
Let us assume that $ r^+(z_*)=0$. By~\eqref{d:r-pm}  one has  
$M_{11}(z_*)=e^{ip(z_*)}\ne e^{-ip(z_*)}$. So, 
\begin{equation*}
  \frac{e^{-ip(z)}-M_{11}(z)}{2i\sin p(z)}= 
\frac{e^{-ip(z)}-e^{ip(z)}}{2i\sin p(z)}+o(1)= -1+o(1)
\quad \text{as}\quad z\to z_*.
\end{equation*}
Therefore~(\ref{eq:poles:1}) implies that, in a neighborhood of $z_*$, 
one has \eqref{eq:omega:near-a-zero-of-b}.
This and~\eqref{d:r-pm} imply that up to a non-vanishing analytic factor
\begin{equation}\label{eq:poles:3}
  \exp\left(\int_{z_0}^z \Omega_+(z)\right)\,  r^+(z)\sim 
  \begin{pmatrix}
    1\\ \frac{e^{ip(z)}-M_{11}(z)}{M_{12}(z)}
  \end{pmatrix},\quad z\sim z_*.
\end{equation}
Now, the analyticity of $z\mapsto \exp\left(\int_{z_0}^z\Omega_+(z)\right)\,  r^+(z)$
follows from the equality
\begin{equation*}
  \frac{e^{ip(z)}-M_{11}(z)}{M_{12}(z)}=-\frac{M_{21}(z)}{e^{-ip(z)}-M_{11}(z)}
\end{equation*}
that is equivalent to~\eqref{eq:poles:2}. The fact that the right hand side in
\eqref{eq:poles:3} does not vanish at $z_*$ is obvious.
This completes the proof of~\Cref{pro:omega:poles}.
\end{proof}
\noindent
Finally, we check
\begin{lemma}\label{le:sum_of_Omegas} In $R$ one has
  \begin{equation}
    \label{eq:sum_of_Omegas}
     \Omega_++\Omega_-=-d \ln\det( r^+\, r^-).
  \end{equation}
\end{lemma}
\begin{proof} 
The statement follows from~\eqref{le:f:Omega} and~(\ref{f:det}).
\end{proof}
\subsection{Proof of \Cref{th:V}}
\label{s:proof:thm:Omegas}
Let $z_0\in R$ and $r_\pm(z_0)\ne 0$. Then, $M_{12}(z_0)\ne 0$, and  
$\det( r^+(z_0)\, r^-(z_0))\ne 0$.
\\
Let us construct $V^\pm$  in terms of $ r^\pm$ defined by~(\ref{d:r-pm}).
As in $R$ the differentials  $\Omega_\pm$ have poles only at zeros of $M_{12}$
(see~\eqref{le:f:Omega}), then, in view of Proposition~\ref{pro:omega:poles}, $V^\pm$ 
are analytic in $R$. As, outside the set of zeros of $M_{12}$, \  $r_\pm$ do not vanish and 
$\Omega_\pm$ are holomorphic, the same proposition implies that  $V^\pm$ 
do not vanish in  $R$. 
\\
Let us check \eqref{eq:det}. Near  $z_0$ the  $\Omega_\pm$ are 
holomorphic, and, using  \Cref{le:sum_of_Omegas}, we get
\begin{equation*}
  \det( V^+(z)\; V^-(z))=
\exp\left(\int_{z_0}^z\Omega_++\Omega_-\right)\det (r^+(z)\,r^-(z))
=\det (r^+(z_0)\;r^-(z_0)).
\end{equation*}
This is formula~\eqref{eq:det}. It is valid  in the whole domain $R$
as $\det(  V^+\; V^+)$ is analytic.
\\
Now, let us  check that in $R$  any   analytic eigenvectors $\tilde V^\pm$ normalized at $z_0$
coincide with $ V^\pm$ up to constant factors. For this, we consider $\tilde r^\pm$, two nontrivial  
solutions to~(\ref{d:rev}) analytic in $R$ and such that ${\rm det}\,(\tilde r^+(z_0)\;\tilde r^-(z_0))\ne 0$. 
In terms of  $\tilde r^\pm$, we  define $\tilde l^\pm$, 
$\tilde \Omega_\pm$ and  $\tilde V^\pm$ as we defined  $l^\pm$, $\Omega_\pm$ 
and  $V^\pm$ in terms of $r^\pm$. One has
\begin{equation}
  \tilde r^\pm(z)=c^\pm(z) r^\pm(z), \quad z \in R,
\end{equation}
where $c^\pm$ are nontrivial functions meromorphic in $R$.
Clearly, $c^\pm$ can vanish only at points where $\tilde r^\pm$ vanish,
and  $c^\pm$ can have poles  only at points where $r^\pm$ vanish.
So, near  $z_0$ the functions $c^\pm$ are analytic, do not vanish, and one has
\begin{equation*}
\int_{z_0}^z \frac{\tilde l^\pm\; d\,\tilde r^\pm}
{\tilde l^\pm \tilde r^\pm}=\int_{z_0}^z \frac{l^\pm \;d\,r^\pm}
{l^\pm\;r^\pm}+\left.\ln c^\pm\right|_{z_0}^{z},
\end{equation*}
\begin{equation*}
e^{ \int_{z_0}^z\tilde \Omega_\pm}\tilde r^\pm(z)=e^{ \int_{z_0}^z \Omega_\pm
-\left.\ln{c^\pm}\right|_{z_0}^{z}} c^\pm(z) r^\pm(z)
=c^\pm(z_0)e^{ \int_{z_0}^z \Omega_\pm}r^\pm(z).
\end{equation*}
This implies that  in $R$, one has $\tilde V^\pm=c^\pm(z_0) V^\pm$, i.e., any  
analytic eigenvectors normalized at $z_0$ coincide with $ V^\pm$ up to 
constant factors. 
\\
Finally, as $\tilde V^\pm(z)=c^\pm(z_0) V^\pm(z)$ and $\tilde r^\pm(z)=c^\pm(z) r^\pm(z)$, 
formula~\eqref{eq:det} valid for $V^\pm$ is valid also for $\tilde V^\pm$.
This  completes the proof of \Cref{th:V}.
\qed
\begin{corollary}[from the proof of \Cref{th:V}]
Let $\tilde r^\pm$ be nontrivial analytic solutions to~\cref{d:rev} and $\tilde  l^\pm$
be constructed by formula~\eqref{eq:left-right}.
Then  $z\mapsto\tilde l^\pm(z) \tilde r^{\pm}(z)$ are nontrivial analytic functions.
\end{corollary}
\begin{proof} The statement follows from~\eqref{le:det} as $\tilde r^{\pm}(z)=
c^{\pm} r^{\pm}$, where $c^\pm$ are nontrivial meromorphic  functions. 
\end{proof} 
\subsection{$\Omega_\pm$ near turning points}
\subsubsection{The complex momentum near a turning point}
Let $z_0\in D$ be a turning point for \cref{main}. 
We call it {\it simple} if $(\text{Tr M})'(z_0)\ne 0$. 
\\
One can easily see that, near a simple turning point $z_0$, the complex momentum 
is an analytic function of  $\tau=\sqrt{z-z_0}$, and one has
\begin{equation}
  \label{eq:p-tau}
  p(z)=p(z_0)+p_1\tau+o(\tau), \quad \tau\to 0,
\end{equation}
where $p_1$ is a non-zero constant.
Below, near a simple branch point $z_0$, we choose $\tau=\sqrt{z-z_0}$ as 
the local coordinate.
\subsubsection{$\Omega_\pm$ near a turning point} 
Let $z_0$ be a simple turning point. One has
\begin{lemma}\label{le:Omega:bp:1}
The point $z_0$ is a simple pole of $\Omega_\pm$.
If $M_{12}(z_0)\ne 0$, then   $\res_{z_0}\Omega_\pm =-\frac12$.
Otherwise, $\res_{z_0}\Omega_\pm =-\frac32$.
\end{lemma}
\begin{proof}
For the sake of definiteness, we prove this lemma only for $\Omega_+$.
\\
First, we consider the case where $M_{12}(z_0)\ne 0$.
Let us consider the terms in the right hand side of~(\ref{le:f:Omega}). In a neighborhood
of $\tau=\sqrt{z-z_0}=0$, one has :
\begin{itemize}
  \item as $p$ is analytic in $\tau$, and thus, $d p$ is  a holomorphic differential;
  \item as $z=z_0+\tau^2$, \ $e^{-ip}-M_{11}$ is analytic in $\tau$; 
  \item as $d z=2\tau\,d\tau$, one has $d M_{11}=\tau g(\tau)\,d\tau$,
 where $g$ is analytic;
 \item as $M_{12}(z_0)\ne 0$,   one has
  $d\ln M_{12}=\tau f(\tau)\,d\tau$, where $f$ is analytic;
  \item as $p(z_0)\in \pi \Z$, and in view of~\eqref{eq:p-tau}, 
$\sin p$ is analytic in $\tau$ and has a simple zero at $\tau=0$. 
\end{itemize}
These observations and formula~\eqref{le:f:Omega} imply that, 
in a neighborhood of $\tau=0$, 
\begin{equation*}
  \Omega_+=-\frac{de^{ip}}{2i\sin p}+\text{a holomorphic differential}.
\end{equation*}
As $\frac{de^{ip}}{2i\sin p}=\frac12\,d\ln(e^{2ip}-1)$, and as $e^{2ip}-1$ has a simple 
zero at $z_0$, see formula~\eqref{eq:p-tau}, this implies that $z_0$ is a simple pole of $\Omega_+$,
and that $\res_{z_0}\Omega_+=-1/2$.
\\
Now, we assume that $M_{12}(z_0)=0$. 
Then, in a neighborhood of $\tau=0$, the differential
\begin{equation}\label{eq:Omega:bp:1:1}
  \Omega_+-\left(\frac{\left(e^{-ip}-M_{11}\right)\, d\,\ln M_{12}}{2i\sin p}-
\frac{de^{ip}}{2i\sin p}\right)
\end{equation}
is holomorphic. Let us consider the first term in the brackets.
\\
We have $e^{ip(z_0)}=e^{-ip(z_0)}$. On the other hand, as $M_{12}(z_0)=0$, 
and as $\det M\equiv 1$, \ $M_{11}(z_0)$ and $M_{22}(z_0)$ are eigenvalues of $M(z)$. 
So, we have $ M_{11}(z_0)=M_{22}(z_0)$.
\\
Using the definition of the complex momentum,  we get
$$e^{-ip}-M_{11}=\frac{e^{-ip}-M_{11}}2+\frac{M_{22}-e^{ip}}2=
-i\sin p+\frac{M_{22}-M_{11}}2.$$
Therefore, near $\tau=0$, one has
\begin{equation}\label{eq:Omega:bp:1:2}
\frac{e^{-ip}-M_{11}}{2i\sin p}=-\frac12+O(\tau).
\end{equation}
Now, to complete the proof, it suffices to check that near $\tau=0$
\begin{equation}\label{eq:Omega:bp:1:3}
d\ln M_{12}=\frac{2\,d\tau}{\tau}+\text{a holomorphic differential}.
\end{equation}
Indeed, this and \eqref{eq:Omega:bp:1:2} imply that
the first term in the brackets in \eqref{eq:Omega:bp:1:1}
has a simple pole with the residue  equal to $-1$. On the other hand,
we have already seen that at $\tau=0$ the second term in the brackets
has a simple pole with the residue equal to $-\frac12$. These 
observations lead to the second statement of the lemma.
\\
As $\frac{2d\,\tau}{\tau}=\frac{dz}{z}$, to prove 
representation~\eqref{eq:Omega:bp:1:3}, we need only to check  
that the zero of $M_{12}$ at $z_0$ is simple. As $\det M\equiv 1$, 
and as $M_{11}(z_0)=M_{22}(z_0)\ne 0$, we have
\begin{equation}
M_{12}'M_{21}|_{z=z_0}=M_{11}'M_{22}+M_{22}' M_{11}-M_{12}M_{21}'|_{z=z_0}=
M_{11}(z_0) (\text{Tr}\, M)'(z_0).  
\end{equation}
So, as $z_0$ is a simple turning point, one has $M_{12}'(z_0)\ne 0$.
The proof  is completed. 
\end{proof}
\subsection{The behavior of $p$ and $\Omega_\pm$ as $|\im z|\to\infty$} 
Below, we assume that $M$ is a trigonometric polynomial satisfying
the assumptions formulated in \cref{ss:assumptions}.
\\
We assume that $Y>0$ is so large that 
the half-planes $\C_u(Y)=\{\im z\ge Y\}$ and $\C_d(Y)=\{\im z\le -Y\}$ are regular,
and $M_{12}$ does not vanish in them.
\\
Here, we study the complex momentum and  $\Omega_\pm$ 
in $\C_u(Y)$ and $\C_d(Y)$. In particular, we get their asymptotic representations
as $|y|\to\infty$, \ $y=\im z$.
\\ 
Below $C$ denotes different positive constants, and $O(f(z))$ denotes 
an expression bounded by $C|f(z)|$ in the domain we consider.
\\
For a trigonometric polynomial $P$,  \ $P(z)=\sum_{j=-l}^k P_j e^{2\pi i j z}$, where
$P_j$ are  Fourier coefficients, and  $P_{-l} P_k\ne 0$,  we let $P_u=P_{-l}$, 
$P_d=P_{k}$,  $n_u(P)=l$ and $n_d(P)=k$.
\\
Let $t=\tr M$. In view our assumptions made in \cref{ss:assumptions}, one has
\begin{equation}\label{nUnA}
n_s (M_{12} ), \ n_s (M_{21}), \ n_s (M_{22} ) \le n_s (M_{11} ) = n_s(t)>0,\qquad s\in\{u,d\}.
\end{equation}
We also note that this and the equality $\det M\equiv 1$ imply that
\begin{equation} \label{nUnA:1}
n_u(M_{22})\le n_u(M_{12}).
\end{equation}
\subsubsection{The behavior of the complex momentum}
Let us fix in $\C_u(Y)$ an analytic branch $p$ of the complex momentum.
In view of~(\ref{eq:p}),  one has
\begin{equation}\label{as:p:up}
  p(z)=s_u\left(2\pi n_u(t) z +i\ln t_u +O(e^{-2\pi y})\right),\quad y=\im z,\quad  z\in \C_u(Y),
\end{equation}
where  $s_u\in\{\pm 1\}$ and the branch of $\ln$ are determined by 
the choice of the branch $p$. We note that by our assumptions $n_u(t)>0$, 
see \cref{ss:assumptions}.
\\
By means of  the Cauchy estimates for the derivatives of analytic functions,
we deduce from~\eqref{as:p:up} the estimates : 
\begin{equation}\label{eq:p-prime}
  p'(z)=2s_u\pi n_u(t)  +O(e^{-2\pi y}), \ 
 p''(z)=O(e^{-2\pi y}),\quad y=\im z,\  z\in \C_u(Y).
\end{equation}
We also note that 
\begin{equation}
  \label{eq:p:periodicity}
  p(z+1)=p(z)+2 s_u \pi n_u(t) ,\quad  z\in\C_u(Y). 
\end{equation}
Indeed, it follows from~\eqref{eq:p}, that $p(\cdot +1)$ is a branch of 
the complex momentum analytic in $\C_u(Y)$. This and \eqref{eq:p} imply 
that $p(\cdot+1)=s\, p(\cdot) \ \mod 2\pi$, \ $s\in\{\pm 1\}$. This and~\eqref{as:p:up} 
imply~\eqref{eq:p:periodicity}.
 
Let us fix in $\C_d(Y)$ an analytic branch $p$ of the complex momentum.  
Reasoning as for $\im z \ge Y$, we now prove that
\begin{equation}\label{as:p:down}
  p(z)=s_d\left(-2\pi n_d(t) z +i\ln t_d +O(e^{-2\pi |y|})\right),\quad  z\in \C_d(Y),
\end{equation}
where $s_d\in\{\pm 1\}$. Note that $n_d(t)>0$. Furthermore, we have
\begin{equation}\label{eq:p-prime:d}
  p'(z)=-2s_dn_d(t)\,\pi   +O(e^{-2\pi y}),
\quad p''(z)=O(e^{-2\pi y}),\quad   z\in \C_d(Y),
\end{equation}
and
\begin{equation}
  \label{eq:p:periodicity:down}
  p(z+1)=p(z)-2 s_d \pi n_d(t),\quad z\in\C_d(Y). 
\end{equation}
\subsubsection{The behavior of $\Omega_\pm$}\label{ss:Omega:infty}
Let $p$ be a branch of the complex momentum analytic  in $\C_u(Y)$  
satisfying~\eqref{as:p:up} with $s_u=1$. Here we study 
in $\C_u(Y)$  the differentials $\Omega_\pm $ defined in terms of this branch $p$ 
by~\eqref{le:f:Omega}.
\\
The half-plane $\C_u(Y)$ being regular, we can represent there $\Omega_\pm$ in the form
$\Omega_\pm(z)=\omega_\pm(z) dz$. For our choice of $Y$, the functions $\omega_\pm$ 
are analytic in $\C_u(Y)$. 
\\
Thanks to~\eqref{eq:p:periodicity}, one has
\begin{equation}
  \label{eq:omega:periodicity}
  \omega_\pm(z+1)=\omega_\pm(z),\quad z\in \C_u(Y).
\end{equation}
Let us check
\begin{proposition}\label{pro:omegas-at-infty:u} For sufficiently large  $Y$
and   $z\in \C_u(Y)$, one has
\begin{equation}\label{as:omegas}
\begin{split}
\omega_+(z)&=\pi i n_u(t)+O(e^{-2\pi y}),\\ 
\omega_-(z)&=\pi i n_u(t)+2\pi i n_u(M_{12})+O(e^{-2\pi y}).
\end{split}
\end{equation}
\end{proposition}
\begin{proof} Let us begin with $\omega_+$.
Using~\eqref{le:f:Omega} and (\ref{eq:p}), we get
\begin{equation*}
\omega_+(z)=- \frac{ip'(z)}2 -\frac{M_{22}(z)\frac {M_{12}'(z)}{M_{12}(z)}+M_{11}'(z)-
e^{ip(z)}\left(\frac {M_{12}'(z)}{M_{12}(z)}+ip'(z) \right)}
{M_{11}(z)+M_{22}(z)-2e^{ip(z)}}.
\end{equation*}
In view of~\eqref{as:p:up} and as $s_u=1$,  one has
\begin{equation}\label{est:exp-and-tr}
  e^{ip(z)}=O(e^{-2\pi y}),\quad  e^{-ip(z)}/\tr M(z)\to 1,\quad y\to \infty.
\end{equation}
This and~\eqref{eq:p-prime} lead to the formula  
\begin{equation}\label{omega:calc:1}
\omega_+(z)=-\pi i n_u(t)-\frac{M_{22}(z)\frac {M_{12}'(z)}{M_{12}(z)}+M_{11}'(z)}
{M_{11}(z)+M_{22}(z)} +O(e^{-2\pi y}).
\end{equation}
Now, we consider the case where $n_u(M_{12})=n_u(M_{11})$.
Then as $z\to \infty$ one has 
$$\frac {M_{12}'(z)}{M_{12}(z)}=-2\pi i n_u(M_{11})+O(e^{-2\pi y}),\quad
\frac {M_{11}'(z)}{M_{11}(z)}=-2\pi i n_u(M_{11})+O(e^{-2\pi y}),$$
and, in view of~\eqref{nUnA}, we get
\begin{equation}\label{omega:calc:2}
  \frac{M_{22}(z)\frac {M_{12}'(z)}{M_{12}(z)}+M_{11}'(z)}
{M_{11}(z)+M_{22}(z)}=-2\pi i n_u(t) +O(e^{-2\pi y}).
\end{equation}
This and~\eqref{omega:calc:1} leads to the first formula in~\eqref{as:omegas}.
\\
To complete the proof, we have to analyze the case where $n_u(M_{12})<n_u(M_{11})$.
Then, in view of~\eqref{nUnA:1}, one has 
$M_{22}(z)/M_{11}(z)=O(e^{-2\pi y})$, and we again come to~\eqref{omega:calc:2},
and thus to  the first formula in~\eqref{as:omegas}. This complete its proof.

Now, let us turn to $\omega_-$. Instead of~\eqref{omega:calc:1}, we now get
\begin{equation}\label{omega:calc:3}
\omega_-(z)=\pi i n_u(t)-\frac{M_{11}(z)\frac {M_{12}'(z)}{M_{12}(z)}+M_{22}'(z)}
{M_{11}(z)+M_{22}(z)} +O(e^{-2\pi y}),
\end{equation}
and considering consequently the case where $n_u(M_{12})=n_u(M_{22})$ and then
the case where $n_u(M_{12})>n_u(M_{22})$ (and, therefore, $n_u(M_{11})>n_u(M_{22})$) 
we prove that   
\begin{equation*}
\frac{M_{11}(z)\frac {M_{12}'(z)}{M_{12}(z)}+M_{22}'(z)}
{M_{11}(z)+M_{22}(z)}=-2\pi i n_u(M_{12}) +O(e^{-2\pi y}).
\end{equation*}
This leads to the second formula in~\eqref{as:omegas}. 
The proof  is complete.
\end{proof}
\noindent
Let $p$ be a branch of the complex momentum analytic  in $\C_d(Y)$  
and satisfying~\eqref{as:p:down} with $s_d=1$. Now, we study 
in $\C_d(Y)$  the  $\Omega_\pm $ defined in terms of this $p$ 
by~\eqref{le:f:Omega}.
\\
One has $\Omega_\pm(z)=\omega_\pm(z) dz$, where $\omega_\pm$ are analytic in $\C_d(Y)$ 
functions. 
\\
We get the formula
\begin{equation}
  \label{eq:omega:periodicity:d}
  \omega_\pm(z+1)=\omega_\pm(z),\quad z\in \C_d(Y),
\end{equation}
and
\begin{proposition}\label{pro:omegas-at-infty:d} Let $Y$ 
be sufficiently large. Then in $C_d(Y)$ 
\begin{equation}\label{as:omegas:d}
\begin{split}
\omega_+(z)&=-\pi i n_d(M_{11})+O(e^{-2\pi |y|}),\\ 
\omega_-(z)&=-\pi i n_d(M_{11})-2\pi i n_d(M_{12})+O(e^{-2\pi |y|}).
\end{split}
\end{equation}
\end{proposition}
\noindent
The proof of this proposition being similar to one of \Cref{pro:omegas-at-infty:u},
we omit it.
\subsection{Remarks on the Riemann surface of $\Omega_\pm$}
The differentials $\Omega_\pm$  are two branches of a meromorphic 
differential $\Omega$ defined on the Riemann surface of the analytic 
function $w:z\mapsto e^{ip(z)}$ (this Riemann surface has two sheets).

As $\tr M$ is a trigonometric polynomial,  it is natural to consider $w$ as a function of 
the variable $u=e^{2\pi i z}$. Then the Riemann surface $\Gamma$ of $w$ appears 
to be a hyperelliptic curve. In particular, in the case where $\tr M$ is a first order 
trigonometric polynomial, relation  \eqref{eq:p} implies that 
\begin{equation}
  w+1/w=t_{1}u +t_0+t_{-1}/u, \quad u\in \C,
\end{equation}
where $t_1$, $t_0$ and $t_{-1}$ are constants, and 
$|t_{1}|^2+|t_{-1}|^2\ne0$.  Therefore $w$ is single-valued  on the Riemann
surface of the function $u\mapsto \sqrt{(t_{1}u^2 +t_0u+t_{-1})^2-4 u^2}$, which 
is a hyperelliptic curve of genus one, see~\cite{springer}. 

The analysis done in the previous sections shows that on $\Gamma$ the differential 
$\Omega$ has simple poles at zeros of $M_{12}$ (on the sheets where $w(z)-M_{11}(z)$ 
vanishes), at all the branch points of $p$,   at zero and at infinity.

The fact that $\Omega$ is meromorphic on a hyperelliptic curve $\Gamma$ is
important for applications of the complex  WKB method. In particular, it implies
that  the differential $\Omega$ can be expressed in terms of standard abelien 
differentials defined on $\Gamma$, and that the integrals of $\Omega$ along  
closed  curves on $\Gamma$ can be expressed in terms of  integrals  
along a finite number of cycles (closed curves) of a canonical basis of the first 
homology group of $\Gamma$. We omit further details and note only that the 
reader can find examples of using the theory of hyperelliptic curves in the  
WKB analysis in \cite{F-K:05a}.  
\section{The proof of \Cref{th:main:local} for bounded 
canonical domains}\label{s:proof1}
We prove \Cref{th:main:local} by reducing the analysis of~\cref{main}
to analyzing  a finite difference equation of precisely the same form 
as the one studied  in \cite{F-Shch:17}. 
\\
Below $R$ is a regular horizontally connected domain, and $p$ is a branch of the
complex momentum analytic in $R$. We always assume that $z,z+h\in R$. 
\\
Also, for a matrix-function $A$, \ $A^{-1}(z)$ is the matrix inverse to $A(z)$.
\\
In this section all the estimates and asymptotics  are locally uniform in $z$. 
\\
We pick $z_0\in R$ so that $\det (r^+(z_0)\;r^-(z_0))\ne 0$, and define in terms 
of $\Omega_\pm$ and $r^\pm$ the analytic eigenvectors $V^\pm$ of $M$ 
normalized at $z_0$. 
\subsection{Asymptotic transformation
of the matrix in~\eqref{main}}  
Let us note that the leading terms in~\eqref{as:main}, i.e., the vectors
\begin{equation}
  \label{eq:Psi0}
  \Psi_0^\pm(z)= e^{\pm\frac{i}{h}\int_{z_0}^zp(z)\,dz}V^\pm(z),
\end{equation}
are eigenvectors of $M(z)$, corresponding to its eigenvalues $e^{\pm ip(z)}$.
In view of~(\ref{eq:det}),
\begin{equation}\label{det:Psi0}
\det(\Psi_0^+(z)\; \Psi_0^-(z)) =\det(r^+(z_0)\;r^-(z_0)).
\end{equation}
We define the matrix 
$\Psi_0(z)=(\Psi_0^+(z)\,\Psi_0^-(z))$ and represent a vector solution $\Psi$ to $\cref{main}$
in the form $\Psi(z)=\Psi_0(z)X(z)$.  Then $X$ satisfies the equation
\begin{equation}\label{eq:Phi}
  X(z+h)=T(z)X(z)
\end{equation}
with
\begin{equation}
\label{eq:TM}
T(z)=\Psi_0^{-1}(z+h)M(z) \Psi_0(z).
\end{equation}
We prove
\begin{proposition}\label{pro:diag} As $h\to 0$, one has
\begin{equation}\label{T:as}
  T(z)=I+
  \begin{pmatrix}
    O(h^2) & O(h)e^{-\frac{2i\theta(z)}h}\\
    O(h)\,e^{\frac{2i\theta(z)}h} & O(h^2)
  \end{pmatrix},
\quad 
\theta(z)=\int_{z_0}^zp(z)\,dz.
\end{equation}
\end{proposition}
\begin{proof}
As $\Psi_0^\pm(z)$ are eigenvectors of $M(z)$ corresponding to its eigenvalues 
$e^{\pm ip(z)}$, 
\begin{equation}\label{eq:TPsi}
T(z)= \Psi_0^{-1}(z+h)\Psi_0(z) 
\begin{pmatrix} e^{ip(z)} & 0\\ 0 & e^{-ip(z)}  \end{pmatrix}.
\end{equation}
In view of (\ref{eq:Psi0}), we have
\begin{equation}
  \label{eq:PsiV}
  \Psi_0(z)=V(z)\begin{pmatrix} e^{i\theta(z)/h} & 0\\ 0 & e^{-i\theta(z)/h}  \end{pmatrix},
\quad V(z)=(V^+(z)\; V^-(z)).
\end{equation}
Formulas~\eqref{eq:TPsi} and  \eqref{eq:PsiV} imply that
\begin{equation}\label{eq:TV}
  T(z)= 
\begin{pmatrix} e^{-\frac{i\theta(z+h)}h} & 0\\ 0 & e^{\frac{i\theta(z+h)}h}  \end{pmatrix}
  W(z)
\begin{pmatrix} e^{\frac{i\theta(z)}h+ip(z)} & 0\\ 0 & e^{-\frac{i\theta(z)}h-ip(z)}  \end{pmatrix},
\end{equation}
where $W(z)=V^{-1}(z+h)V(z)$.
To continue, we need
\begin{lemma} \label{le:V1V}As $h\to 0$, one has
\begin{equation}\label{eq:V1V}
  W(z)=
\begin{pmatrix}e^{ ihp'/2+O(h^2)}& O(h)\\O(h)&  e^{- ihp'/2+O(h^2)}\end{pmatrix}.
\end{equation}
\end{lemma}
\begin{proof}  Using the Taylor's theorem, we get
\begin{equation*}
 W(z)= V^{-1}(z+h)V(z)=I+h(V^{-1})'(z)V(z)+O(h^2).
\end{equation*}
As $(V^{-1}V)'=0$, one has $(V^{-1})'V=-V^{-1}V'$, and
\begin{equation*}
  W(z)=I-hw(z)+O(h^2), \quad w(z)=V^{-1}(z)V'(z).
\end{equation*}
It suffices to check that
\begin{equation}
  \label{as:w}
  w_{11}(z)=-ip'(z)/2,\qquad  w_{22}(z)=ip'(z)/2.
\end{equation}
Let us prove the first formula. 
\\
Let $e^\pm(z)=e^{\int_{z_0}^z \Omega_\pm}$.  One has 
\begin{equation}\label{eq:V}
  V(z)=(r^+(z)e^+(z)\ \  r^-(z)e^-(z))=(r^+(z)\ \  r^-(z))
\begin{pmatrix} e^+(z) & 0\\ 0 & e^-(z)  \end{pmatrix}.
\end{equation}
Therefore, in view of~\Cref{cor:det}, 
\begin{equation}\label{eq:V-inverse}
  V^{-1}(z)=\frac1{\det (r^+(z)\; r^-(z))}
\begin{pmatrix} 1/e^+(z) & 0\\ 0 & 1/e^-(z)  \end{pmatrix}\,
\begin{pmatrix} l^+(z)\\ -l^-(z)  \end{pmatrix},
\end{equation}
and using~\eqref{eq:det:1}, we get finally
\begin{equation*}
  V^{-1}(z)=\frac1{l^+(z)r^+(z)}
\begin{pmatrix} l^+(z)/e^+(z)\\ -l^-(z)/e^-(z)  \end{pmatrix}.
\end{equation*}
Therefore,
\begin{equation*}
 w_{11}(z)=\frac1{l^+(z) r^+(z)}\;\frac{l^+(z)}{e^+(z)}\;(r^+e^+)'(z),
\end{equation*}
and using the definition of $\Omega_+$, see~\eqref{def:omega},
we get 
\begin{equation*}
  w_{11}(z)=\frac{(e^+)'(z)}{e^+(z)}+\frac{l^+(z) (r^+)'(z)}{l^+(z) r^+(z)}=
-ip'(z)/2.
\end{equation*}
This proves the first formula in~\eqref{as:w}.  The second one is checked similarly.
\end{proof}
\noindent
As $\theta(z+h)=\theta(z)+p(z) h+p'(z) h^2/2+O(h^3)$, substituting
representation~\eqref{eq:V1V} into
formula \eqref{eq:TV}, we come to~\eqref{T:as}.
This completes the proof of \Cref{pro:diag}.
\end{proof}
\subsection{Solutions to equations \eqref{eq:Phi} 
and \eqref{main}}
\Cref{eq:Phi} with a matrix $T$ of the form~\eqref{T:as}
is precisely the equation we study in~\cite{F-Shch:17}, see the beginning 
of section 4  and Lemma 4.1   in~\cite{F-Shch:17}. Most of~\cite{F-Shch:17} \ 
(sections 4--6) is devoted to the analysis of this equation.
The results of this analysis are described  as properties  of a vector-function $\tilde X$
defined  in terms of $X$ by formulas (5.1) and (5.2) in \cite{F-Shch:17}. 
Below, we describe these results as properties of $X$.
\\
In~\cite{F-Shch:17}, in the formula analogous to \eqref{T:as}, \ $p$ is a function 
analytic in a regular domain $R$, and, in terms of this function $p$, one defines 
the canonical domains  exactly as in \Cref{sss:CD}. Then one proves that, given a 
bounded canonical domain $K\subset R$, for sufficiently small $h$, there  exist two solutions 
to \cref{eq:Phi} that are analytic in $K$ and admit there as $h\to 0$ the asymptotic 
representations
\begin{equation*}
  X^+(z)=\begin{pmatrix}1\\0\end{pmatrix}
        +\begin{pmatrix}O(h)\\ e^{\frac{2i\theta(z)}h}O(h)\end{pmatrix},\quad\text{and}\quad
 X^-(z)=\begin{pmatrix}0\\1\end{pmatrix}
        +\begin{pmatrix}e^{-\frac{2i\theta(z)}h}O(h)\\O(h) \end{pmatrix}.
\end{equation*}
The representation for $X^+$ follows from Lemma 5.1 in \cite{F-Shch:17}, the 
representation for $X^-$ is obtained as described in section 6.3 in \cite{F-Shch:17}. 
We omit further details and note only that  $\tilde X^\pm$ satisfy  singular integral 
equations on a vertical curve $\gamma$, and that  the crucial observation is that  
if $\gamma$ is a canonical curve, then 
the norms of the integral operators are small. 
\\
Having constructed $X^\pm$, one constructs the solutions $\Psi^\pm$ from 
\Cref{th:main:local} by the formulas $\Psi^\pm(z)=\Psi_0(z)X^\pm(z)$.
This completes the proof of \Cref{th:main:local}  for bounded canonical domains.
\qed
\section{The proofs of Theorems~\ref{th:main:local} 
and~\ref{th:main:global} for unbounded canonical 
domains}\label{s:proof2}
In \cite{F-Shch:18} we studied the one-dimensional difference 
Schr\"odinger equations with the potentials  being trigonometric 
polynomials. Now, we consider \cref{main} with $M$ being a trigonometric
polynomial and prove Theorems~\ref{th:main:local} 
and~\ref{th:main:global} for unbounded canonical 
domains by means of the method developed in \cite{F-Shch:18}.
\\
Again, $R$ is a regular horizontally connected domain, and $p$ is  
a branch of the complex momentum analytic in $R$. As before 
$V^\pm$ are normalized at $z_0\in R$.
\\
Now we assume that the domain $R$ contains an infinite vertical 
curve and is horizontally bounded.  Also, almost up to the 
end of this section, we assume  that
\begin{equation}\label{hyp:exp}
e^{ip(z)}\to 0\quad\text{as}\quad |\im z |\to\infty,
\end{equation}
i.e., that the coefficients $s_u$ and $s_d$ in~\eqref{as:p:up} 
and~\eqref{as:p:down} are equal  to $+1$. 
\\
Finally, $C$ denotes different positive constants independent of $h$,
and, for $z$ being in the domains we consider,  $O(f(z,h))$ is bounded by $C|f(z,h)|$.
\subsection{Asymptotic transformation
of the matrix in~\eqref{main}}  
We begin with transforming  \cref{main} as in the previous section.
\subsubsection{The matrix $W(z)$}
The statements of \Cref{le:V1V} remain valid in any compact subset  of $R$. 
Now, we assume that $Y$ is so large that $\C_u(Y)\cup \C_d(Y)$ is regular and $M_{12}(z)\ne 0$  in 
$\C_u(Y)\cup\C_d(Y)$. Then we continue analytically the functions $p$, $r^\pm$, $V^\pm$ 
and $W$ in $\C_u(Y)\cup\C_d(Y)$ from $R$, and prove 
\begin{lemma} \label{le:W} 
Let $s\in\{d,u\}$. If  $Y$ is sufficiently large,  
then, for  $z\in \C_s(Y)$ one has
\begin{equation}\label{eq:W:diag}
W_{11}(z)=e^{ ihp'/2+h^2g_1(z)},\qquad 
W_{22}(z)=e^{ -ihp'/2+h^2g_2(z)},
\end{equation}
\begin{equation} \label{eq:W:diag:1}
g_1(z)=c_s+O(e^{-2\pi |y|}), \quad g_2(z)=-c_s+O(e^{-2\pi |y|}),
\end{equation}
\begin{equation}
\label{eq:W:anti}
W_{12}(z)=O(he^{-2\pi n_s(M_{12})\,|y|-2\pi |y|}),\quad   
W_{21}(z)=O(he^{2\pi  n_s(M_{12}) |y|}).
\end{equation}
Here $c_s$ is a constant analytic in $h$.
If $M_{22}(z)/M_{11}(z)\to 0$ as $|y|\to\infty$, then $c_s=0$.
\end{lemma}
\begin{proof}  
Below, we assume that  $z\in \C_u(Y)$; the case of $z\in\C_d(Y)$ is treated similarly. 
We use notations from the proof of~\Cref{le:V1V}.
The analysis is broken into several steps.
We begin with studying $W_{11}$.  
\\
{\bf 1.} \ Let $d_0=\det (r^+(z_0)\; r^-(z_0))$.
By~\eqref{eq:V}, \eqref{eq:V-inverse} and~\eqref{eq:det},
\begin{equation}\label{eq:W11}
  W_{11}(z)=\frac1{d_0}\,e^-(z+h)e^+(z)\, l^+(z+h)r^+(z).
\end{equation}
{\bf 2.} \ Using  \eqref{as:omegas}, we get  
\begin{equation}\label{eq:W11:0}
 e^-(z+h)\,e^+(z)=e^{2\pi i (n_u(M_{11})+ n_u(M_{12})) z+a_0+a_1 h+O(e^{-2\pi y}) },
\end{equation}
where $a_0$ and $a_1$ are constants independent of $h$. Therefore,
\begin{equation}\label{eq:W11:1}
e^-(z+h)\,e^+(z)= O(e^{-2\pi  (n_u(M_{11})+ n_u(M_{12})) y} ).  
\end{equation}
We also note that in, view of~\eqref{eq:omega:periodicity} and~\eqref{eq:W11:0}, 
$e^-(\cdot+h)\,e^+(\cdot)$ is $1$-periodic.
\\
{\bf 3.} \ By means of~\eqref{d:r-pm},~\eqref{d:l-pm} and~(\ref{eq:p}), we check that
\begin{equation}\label{eq:W11:2}
 l^+(z+h)\,r^+(z)=(M_{22}(z+h)-e^{ip(z+h)}) M_{12}(z)+(M_{11}(z)-e^{ip(z)}) M_{12}(z+h).
\end{equation}
Estimates~\eqref{eq:d/a} and~\eqref{est:exp-and-tr}, and formula~\eqref{eq:W11:2}
imply that 
\begin{equation}\label{eq:W11:3}
 l^+(z+h)\,r^+(z)=O(e^{2\pi (n_u(M_{11})+n_u(M_{12}))y}).
\end{equation}
{\bf 4.} \ We have chosen $Y$ so that  $W_{11}$ is analytic in $\C_u(Y)$.  
Moreover, in view of  \eqref{eq:W11:3},  \eqref{eq:W11:1}, 
and \eqref{eq:W11}, it is bounded there.
\\
{\bf 5.} \ The element $W_{11}$ is $1$-periodic in $z\in \C_u(Y)$.
Indeed, by the first step the product $e^+(z+h)e^-(z)$ is $1$-periodic in $z$, 
and \eqref{eq:p:periodicity} and the $1$-periodicity of $M$ imply the  
$1$-periodicity of $l^\pm$ and $r_\pm$.  This and \eqref{eq:W11} imply 
the needed.
\\
{\bf 6.} \  Now we prove the representation for $W_{11}$ from \eqref{eq:W:diag}.
As $\C_u(Y)$ is regular, $p$ is analytic there. By~\eqref{eq:p:periodicity} 
it is $1$-periodic in $\C_u(Y)$. This and the previous two steps imply that
the function $z\to W_{11}(z)-e^{ ihp'(z)/2}$ is  $1$-periodic and analytic  in  $\C_u(Y)$.  
Let us consider it as a function of $u=e^{2\pi i z}$, and denote this new 
function by $f$. In the disk $D=\{|u|\le e^{-2\pi Y}\}$, \ $f$ is  analytic and bounded.
By~\Cref{le:V1V}, on its boundary  one has
$f(u)=O(h^2)$ uniformly in $u$.   This and the Maximum principle for analytic 
functions imply that $f(u)=O(h^2)$ uniformly in $u\in D$. Using the Maximum modulus 
principle again, we see that  $ \frac{f(u)-f(0)}{u}=O(h^2e^{2\pi Y})$ uniformly in  $u\in D$
(as this estimate holds on the boundary of  $D$). Returning to $z$, we get
\begin{equation*}
  W_{11}(z)=e^{ ihp'(z)/2}+h^2c+O(h^2e^{-2\pi( y-Y)}),\quad c=f(0)/h^2, \quad z\in\C_u(Y).
\end{equation*}
As $f(0)=O(h^2)$, $p'$ satisfies \eqref{eq:p-prime:d}, and $Y$ is a fixed positive number, 
this implies the representation for $W_{11}$ from \eqref{eq:W:diag}.
\\
{\bf 7.} \ Let us  assume that $M_{22}(z)/M_{11}(z)\to 0$ as $|y|\to\infty$, 
and prove that $c_u\equiv0$.
\\
Now, instead of~\eqref{eq:W11:3}, for sufficiently large $Y$, for $z\in C_u(Y)$, we get 
\begin{equation*}
 l^+(z+h)r^+(z)=(M_{11})_u(M_{12})_u\,e^{-2\pi i (n_u(M_{11})+n_u(M_{12}))z
-2\pi i n_u(M_{12})h+O(e^{-2\pi y})}.
\end{equation*}
Combining this with~\eqref{eq:W11:0}, we see that
$W_{11}(z)= e^{ \tilde a_0+\tilde a_1 h + O(e^{-2\pi y})}$
with some  constants $\tilde a_0$ and $\tilde a_1$ independent of $h$.
In view of~\eqref{eq:p-prime}, this representation implies 
that the constant $c_u$ in the formula for $W_{11}$ in \eqref{eq:W:diag:1} is zero.
\\
The proof of the statements of \Cref{le:W} concerning $W_{11}$
is completed.
\\
{\bf  8.} \ Let us turn to $W_{12}$ and $W_{21}$.
Instead of~\eqref{eq:W11}, (\ref{eq:W11:1}) and \eqref{eq:W11:2} we get
\begin{gather}
\label{eq:W12}
  W_{12}(z)=\frac1{d_0}\,e^-(z+h)e^-(z)\, l^+(z+h)r^-(z),\\
\label{eq:W12:1}
 e^-(z+h)\,e^-(z)=O(e^{-2\pi  (n_u(M_{11})+ 2n_u(M_{12})) y})\\
\label{eq:W12:2}
\begin{split}
 l^+(z+h)\,r^-(z)&=(M_{22}(z+h)-e^{ip(z+h)}) M_{12}(z)\\
&\hspace{2cm}-(M_{22}(z)-e^{ip(z)}) M_{12}(z+h).
\end{split}
\end{gather}
Let us assume that $e^{ip(z)}=o(M_{22}(z))$ 
as $y\to\infty$. If $Y$ is sufficiently large, then
\begin{gather}
\label{eq:M12}
M_{12}(z)=(M_{12})_u e^{-2\pi i n_u(M_{12}) z +f_{1}(z)},\quad 
f_{1}(z)=O(e^{-2\pi y}),\\
\label{eq:M22}
M_{22}(z)-e^{ip(z)}=(M_{22})_u e^{-2\pi i n_u(M_{22}) z +f_{2}(z)},\quad 
f_{2}(z)=O(e^{-2\pi y}),
\end{gather}
where $f_{1}$ and $f_{2}$ are analytic in $z$.
These formulas and~\eqref{eq:W12:2} imply that 
\begin{equation}\label{eq:W12:3}
\begin{split}
l^+(z+h)\,r^-(z)&=(M_{22})_u (M_{12})_u \;
e^{-2\pi i (n_u(M_{22})+n_u(M_{12})) z  +f_{1}(z)+f_{2}(z)}\times\\
&\hspace{-.5cm}\times(e^{-2\pi i n_u(M_{22})h+f_{2}(z+h)-f_{2}(z)}-
e^{-2\pi i n_u(M_{12})h+ f_{1}(z+h)-f_{1}(z)}).
\end{split}
\end{equation}
Using the Cauchy estimates for the derivatives of analytic functions, in $\C_u(Y)$  
(possibly with a larger  $Y$) we get 
\begin{equation}\label{est:f}
  f_{j}(z+h)-f_{j}(z)=O(he^{-2\pi y}),\quad j\in\{1,2\}.
\end{equation}
Therefore, $$  l^+(z+h)\,r^-(z)=O(he^{2\pi  (n_u(M_{22})+n_u(M_{12})) y}).$$
This and~\eqref{eq:W12:1} imply that 
\begin{equation*}
    W_{12}(z)=O(h e^{-2\pi (n_u(M_{11})-n_u(M_{22})+n_u(M_{12}))y}) \text{ if }
e^{ip(z)}=o(M_{22}(z)) \text{ as }y\to\infty.
\end{equation*}
Let $n_u(M_{22})=n_u(M_{11})=n_u(t)$. Then, as  $\det M\equiv 1$, we have $n_u(M_{12})=n_u(M_{22})$.
So, now the expression in the brackets in \eqref{eq:W12:3} equals 
$$e^{-2\pi i n_u(M_{12})h}\,(e^{f_{2}(z+h)-f_{2}(z)}-e^{f_{1}(z+h)-f_{1}(z)})=O(e^{-2\pi y}),$$
and, therefore,
\begin{equation*}
  W_{12}(z)=O(he^{-2\pi  n_u(M_{12}) y -2\pi y})\quad \text{if}\quad n_u(M_{22})=n_u(M_{11}).
\end{equation*}
Finally, if $M_{22}(z)=O(e^{ip(z)})$, then,  using~\eqref{as:p:up}, 
we get 
\begin{equation}\label{eq:W12:4}
W_{12}(z)=O(h e^{-2\pi (2n_u(M_{11})+n_u(M_{12}))y}).
\end{equation}
The obtained estimates lead to the estimate for $W_{12}$ from~\eqref{eq:W:anti}.
Similarly one proves the estimate for  $W_{21}$. 
\\
{\bf 9.} \ The representation for $W_{22}(z)$ from~(\ref{eq:W:diag})
follows from the representations for the other elements of $M(z)$ and the 
relation $\det M(z)\equiv 1$. 
\\
The proof of \Cref{le:W} is complete.
\end{proof}
\subsubsection{The matrix $T$}
Let us recall that $T(z)$ is described  by~\eqref{eq:TV}.
In addition to \Cref{pro:diag}, we now get
\begin{proposition} \label{pro:T:Y}
Let $s$ be either $u$ or $d$.
For sufficiently large $Y$, in $\C_s(Y)$,
\begin{equation}\label{eq:T:diag}
T_{11}(z)=e^{ h^2c_s  +O(h^2e^{-2\pi |y|}) },\qquad 
T_{22}(z)=e^{ -h^2c_s+O(h^2e^{-2\pi |y|})},
\end{equation}
\begin{equation}\label{eq:T:anti}
\begin{split}
T_{12}(z)&=e^{- \frac{2i\theta(z)}h} O(h\,e^{+2\pi (2n_s(t)- n_s(M_{12}))\,|y|-2\pi |y|}),\\
T_{21}(z)&=e^{+\frac{2i\theta(z)}h} O(h\,e^{ -2\pi (2n_s(t)- n_s(M_{12}))\,|y|}).
\end{split}
\end{equation}
\end{proposition}
\begin{proof} In view of~\eqref{eq:p-prime} and~\eqref{eq:p-prime:d},   in $\C_s(Y)\cap R$, one has
\begin{equation}\label{eq:theta(z+h)}
\theta(z+h)=\theta(z)+p(z) h+p'(z) h^2/2+O(h^3e^{-2\pi |y|}).
\end{equation}
This and~\eqref{eq:W:diag} lead to~\eqref{eq:T:diag}. Moreover, 
using~\eqref{eq:W:anti},  we get the estimates
\begin{equation*}
T_{12}(z)=e^{- \frac{2i\theta(z)}h} O(e^{ -2ip(z)}W_{12}(z)),\quad
T_{21}(z)=e^{+\frac{2i\theta(z)}h} O(e^{+2i p(z)}W_{21}(z)).
\end{equation*}
This and representations~\eqref{as:p:up} 
and~\eqref{as:p:down} with $s_u=s_d=1$ lead to~\eqref{eq:T:anti}.
\end{proof}
\subsubsection{Completing the asymptotic transformation}\label{sss:S}
To use the method developed in~\cite{F-Shch:18}, for sufficiently small $h$ 
we  transform \cref{eq:Phi} to the form
\begin{equation}\label{eq:withS}
  {\mathcal X}(z+h)-{\mathcal X}(z)=S(z) {\mathcal X}(z), \quad
  S(z)=\begin{pmatrix}  0 & S_{12}(z)\\ S_{21}(z)&0 \end{pmatrix}.
\end{equation}
To do this, for $X$, a solution to \cref{eq:Phi}, we set
\begin{equation}\label{XmathcalX}
   \mathcal{X}(z)= \begin{pmatrix} e^{-\phi_1(z)} & 0 \\ 
0 & e^{-\phi_2(z)}\end{pmatrix}\, X(z).
\end{equation}
If
\begin{equation}\label{eq:phis}
  \phi_1(z+h)=\ln T_{11} (z)+\phi_1(z)\quad\text{and}\quad 
\phi_2(z+h)=\ln T_{22} (z)+\phi_2(z),
\end{equation}
then $\mathcal X$ satisfies~\cref{eq:withS} with 
\begin{equation*}
S_{12}(z)=e^{-\phi_1(z+h)+\phi_2(z)}T_{12}(z),\quad\quad 
S_{21}(z)=e^{-\phi_2(z+h)+\phi_1(z)}T_{21}(z).
\end{equation*}
To construct solutions to  \eqref{eq:phis}, we need
\begin{definition}\label{def:str_vert}
We call a vertical curve $\gamma$ {\it strictly} vertical if the 
angles between $\gamma$ and $\R$ at all the points $z\in \gamma$ 
are uniformly bounded away from zero.
\end{definition}
\noindent
First, we assume that  $R$ contains a strictly vertical curve $\gamma$ with some its 
$\delta$-neighborhood $V_\delta$ and its boundary $\partial V_\delta$.  
\\
Next, for each $j\in\{1,2\}$, we fix a branch of $ \ln T_{jj}$  in the corresponding 
equation in~\eqref{eq:phis}.  For this, we choose $Y$ as in \Cref{pro:T:Y}.
In view of Propositions \ref{pro:diag}  and \ref{pro:T:Y}, for sufficiently small $h$, 
we can  choose and choose the branch of $\ln T_{jj}$  that is analytic and equals $O(h^2)$ 
in $V_\delta\cup \C_u\cup\C_d$.  Then, we prove
\begin{lemma}\label{le:phi}    For  sufficiently small $h$, there exist functions 
$\phi_1$ and $\phi_2$ analytic  in $V_\delta\cup \C_u\cup\C_d$ and satisfying there 
the corresponding equations in~\eqref{eq:phis}. Moreover, in $V_\delta$ one has
\begin{equation}\label{est:phis}
 |\phi_j(z)|\le Ch (1+|y|).
\end{equation}
In~\eqref{est:phis} the right hand side can be replaced by  $Ch$ 
if $c_u=c_d=0$. 
\end{lemma}
\begin{proof} Below we assume that  $h$ is sufficiently small.
We fix $j\in\{1,2\}$. To construct  $\phi_j$, a solution to the corresponding 
equation in~\eqref{eq:phis},   we use a known construction for a solution to 
a first order difference equation, see,  e.g., section 3.5 in~\cite{B-F:94}.
\\
For $z\in V_\delta$, we denote by $\gamma(z)$ the curve
containing $z$ and obtained from $\gamma$ by translation. 
Clearly, $\gamma(z)$  is a strictly vertical, and $\gamma(z)
\subset V_\delta$. 
\\
Let $l_j(\cdot )=\ln T_{jj}(\cdot -h/2)$, where $\ln T_{jj}$ is the branch we have chosen
just before formulating the lemma.  If $h$ is sufficiently small, $l_j$ is analytic in 
$V_\delta\cup \C_d(Y)\cup \C_u(Y)$ and equals $O(h^2)$ there. We fix  $z_0\in V_\delta$  
and  let
\begin{equation}
  \label{formula:phij}
 \phi_j(z)=\frac{\pi}{2ih^2} \int_{\gamma(z)}\frac{\int_{z_0}^{\zeta}l_j(t)\,dt}
{\cos^2\frac{\pi(z-\zeta)}{h}}\,d\zeta,\quad z\in V_\delta.
\end{equation}
The fact that the integral in~\eqref{formula:phij} converges and 
defines an analytic function follows from the estimate
\begin{equation}\label{eq:int-l}
  \int_{z_0}^{z}l_j(t)\,dt= O((1+|y| ) h^2),\quad z\in V_\delta .
\end{equation}
We note that if $c_u=c_d=0$, then the  right hand side in
\eqref{eq:int-l} can be replaced by $O(h^2)$  (see~\eqref{eq:T:diag}).
\\
The fact that, for $z,z+h\in V_\delta$, \  $\phi_j$ given by~\eqref{formula:phij} 
satisfies~\eqref{eq:phis} follows from the residue theorem.
\\
Having constructed $\phi_j$ in $V_\delta$, one continues it analytically in $\C_d\cup\C_u$ 
just by means of the corresponding equation from~\eqref{eq:phis}.
\\
Finally, \eqref{est:phis} follows from the estimate
\begin{equation}\label{est:phij:1}
 |\phi_j(z)|\le C \int_{-\infty}^\infty\frac{1+|\eta+y|}{\cosh^2\frac{\pi\eta}{h}}\,d\eta.
\end{equation}
If $z\in \gamma$, \eqref{est:phij:1} follows from~\eqref{eq:int-l}, \eqref{formula:phij} 
and the fact that $\gamma$ is strictly vertical. If $z\not\in\gamma$, then one also uses 
the fact that  $\gamma(z)$ is obtained from $\gamma$ by a translation.
\\
If  $c_u=c_d=0$, then, instead of~\eqref{est:phij:1}, one obtains the estimate
$ |\phi_j(z)|\le C \int_{-\infty}^\infty\frac{d\eta}{\cosh^2\frac{\pi\eta}{h}}$, and it
implies~\eqref{est:phis} with  the right hand side replaced by  $Ch$.
\end{proof}
\noindent 
Let $\phi_1$ and $\phi_2$ in~\eqref{XmathcalX} be the functions from
Lemma~\ref{le:phi}. Then,  for sufficiently small $h$, \Cref{le:phi} and 
Propositions~\ref{pro:diag} and \ref{pro:T:Y}
imply that, in $V_\delta\cup \C_d\cup C_u$, the coefficients $S_{12}$ and $S_{21}$ 
from~\eqref{eq:withS}  are analytic and admit the representations
\begin{equation}\label{eq:S}
S_{12}(z)=he^{- \frac{2i\theta(z)}h} g_{12}\quad\text{and}\quad
S_{21}(z)=he^{+\frac{2i\theta(z)}h} g_{21}
\end{equation}
\begin{equation}\label{est:s}
\begin{split}
g_{12}(z)&=O(e^{+2\pi\, \left(2n_s(t)- n_s(M_{12})-1+c\,h\right)\,|y|}),\\
g_{21}(z)&=O(e^{ -2\pi \,\left(2n_s(t)- n_s(M_{12})-c\,h\right)\,|y|}),
\end{split}
\end{equation}
for $z\in \C_s(0)\cap V_\delta$, \ $s \in \{u,d\}$.  Here  $c=0$ if $c_u=c_d=0$. 
\subsubsection{Completing the proofs Theorems~\ref{th:main:local} 
and~\ref{th:main:global} for unbounded canonical 
domains}
\label{sss:completing}
Equation~\eqref{eq:withS} with $S_{12}$ and $S_{21}$ of the form \eqref{eq:S} 
was studied  in~\cite{F-Shch:18}, compare~\eqref{eq:withS} and~\eqref{eq:S} 
with (2.12) and (2.15) from~\cite{F-Shch:18}. Actually, one can use the method  
of~\cite{F-Shch:18} if there are positive constants  $C$,  $C_1$  and $C_2$ 
(independent of $h$)  such that 
\begin{equation*}
  |g_{12}(z)|\le C e^{C_1|y|}, \quad  |g_{21}(z)|\le C e^{C_1|y|},\quad
|g_{12}(z)g_{21}(z)|\le C e^{-C_2|y|},
\end{equation*}
and this occurs in our case. So, we construct analytic solutions to 
\eqref{eq:withS} as in~\cite{F-Shch:18},  focusing only on the  
modifications. 
\\
In view of Lemmas 3.3 from~\cite{F-Shch:18},  in the case when $\tr M$
is a trigonometric polynomial, any unbounded canonical domain $D$ (containing an
infinite vertical curve) can be extended to a canonical domain $\tilde D$ such that 
$D\subset \tilde D$, and that, $\forall z\in\tilde D$, there is a strictly canonical curve  
containing $z$ and contained in $\tilde D$ with some its $\delta$-neighborhood. 
Here, all the canonical curves and domains are canonical with respect to one and 
the same branch of the complex momentum. 
\\
Clearly, it suffices to prove \Cref{th:main:global} only for the extended canonical 
domains. We assume that $K=R$ is such an extended canonical domain and that 
it is canonical with respect to the branch $p$ fixed above in $R$.
\\
We prove \Cref{th:main:global} in several steps.\\
Below all the canonical curves are  in $K$ and are canonical  
with respect to  $p$.
\\
{\bf 1.} \ Let $\gamma$ be a strictly vertical curve contained in $K$ together with 
some its $\delta$-neighborhood  $V_\delta$ and its boundary  $\partial V_\delta$. Let 
us consider the matrix $S$ constructed in $V_\delta$ as in \cref{sss:S}. 
We shall study the integral equation 
\begin{equation}\label{eq:3} 
\mathcal{X}=\begin{pmatrix}1\\ 0\end{pmatrix}+{\mathcal L}_+\,(S\mathcal{X}) 
\end{equation}
where $\mathcal{L_+}$ is the singular integral operator  acting by the formula
\begin{equation}\label{eq:operator-L}   
\mathcal{L_+}g\,(z)=
\frac1{2i h}\,\int\limits_{\gamma}
\left(\cot\left[\frac{\pi(\zeta-z-0)}h\right]-i\right)\,
g\,(\zeta)\,d\zeta,
\end{equation}
on a suitable space of functions defined on $\gamma$. 
We note that, formally, equation~\eqref{eq:3} can be obtained from~\eqref{eq:withS} 
by inverting the difference operator  in the left hand side of~\eqref{eq:withS}, 
see section 7.1 in~\cite{F-Shch:18}.
\\
The matrix $S$ being anti-diagonal, we readily  deduce from~\eqref{eq:3} an 
equation for the first element of the vector $\mathcal X$. In view of~\eqref{eq:S}, 
it can be written in the form
\begin{equation}\label{eq:X10}
\mathcal{X}_1=1+h^2\mathcal{L_+}
\left(\,g_{12}\mathcal{K_+}\left(\,g_{21}\,\mathcal{X}_1\right)\,\right),
\end{equation}
\begin{equation}\label{eq:operator-K}
\mathcal{K_+}f\,(z)=e^{ -\frac{2i\theta\,(z)}h}\,\mathcal{L_+}\,
  \left(e^{ \frac{2i\theta}h}\,f\right)\,(z).
\end{equation}
{\bf 2.} \ To study~\eqref{eq:X10}, we fix $a\in(0,1)$ and define
\begin{equation}\label{eq:Pi-gamma-a}
\Pi_{\gamma,a}= \left\{ z\in \C ~ : ~ \exists \zeta \in \gamma : 
\im \zeta = \im z  \text{ and } |\re \zeta -\re z|<ah \right\}.  
\end{equation}
Furthermore, let $b,c>0$ and
\begin{equation}
  \label{eq:555}
  \rho(z,b,c)=
  \begin{cases}
    e^{+2\pi b y} \text{ if } y\ge 0,\\
    e^{-2\pi c y} \text{ if } y\le 0.
  \end{cases}
\end{equation}
Let  $H_{\gamma,a,b,c}$ be the  space of functions $f$
analytic in  $\Pi_{\gamma,a}$ and such that the numbers
\begin{equation}
  \label{eq:5}
  \|f\|_{\gamma,a,b,c}=\sup_{z\in \Pi_{\gamma,a}}|\rho(z,b,c)\,f\,(z)|
\end{equation}
are finite. This is a Banach space with the norm $\|\cdot\|_{\gamma,a,b,c}$.
\\
In addition to~\Cref{def:str_vert}, we need 
\begin{definition}
We call a canonical curve $\gamma$ {\it strictly} canonical if it is strictly vertical, and 
the derivatives in~\eqref{def:can} are bounded away from zero uniformly in $y\in\R$.
\end{definition}
\noindent We recall that $C$ denote different positive constants independent of $h$. One has
\begin{proposition}\label{pro:LK} 
Let $\alpha$ be an infinite strictly vertical curve, and let $b,c>0$.  Then
$$\|\mathcal{L}_+\|_{H_{\alpha,a,b,c}\mapsto H_{\alpha,a,0,0}}\le C/h.$$
Let $\alpha$ be an infinite strictly canonical curve located in $K$ with some its $\delta$-neighborhood,
and let $b,c\in \R$. Then,  for sufficiently small $h$,
$$\|\mathcal{K}_+\|_{H_{\alpha,a,b,c}\mapsto H_{\alpha,a,b-2an_u(t),c-2an_d(t)}}\le C.$$
\end{proposition}
\noindent
Mutandis mutatis, the first and the second statements are proved respectively 
as  Propositions 7.1  and 7.2 from~\cite{F-Shch:18}.
\\
Below we assume that that the curve $\gamma$ is strictly canonical.
\\
Let us fix $a$ so that $0<a<\min\left\{\frac1{2n_u(t)},\,\frac1{2n_d(t)}\right\}$.
\Cref{pro:LK} and estimates~(\ref{est:s}) imply that, for sufficiently small $h$, 
there exists $\mathcal{X}_1$, a solution to \eqref{eq:X10} such that 
$\|\mathcal{X}_1-1\|_{\gamma,a,0,0}\le Ch$, i.e., that $\forall z\in \Pi_{\gamma,a}$
\begin{equation}\label{est:X1}
|\mathcal{X}_1(z)-1|\le Ch.
\end{equation}
{\bf 3.} \ We define the function $\mathcal X_2$  in $\Pi_{\gamma,a}$ by the formula
\begin{equation}
  \label{eq:X2}
\mathcal X_2=h\mathcal{L}_+(e^{\frac{2i\theta}h}g_{21}\mathcal X_1).
\end{equation}
As $\mathcal{L}_+(e^{\frac{2i\theta}h}g_{21}\mathcal X_1)=
e^{\frac{2i\theta}h}\mathcal{K}_+(g_{21}\mathcal X_1)$,
Proposition~\ref{pro:LK} and estimates~\eqref{est:s} and~\eqref{est:X1} 
imply that, for sufficiently small $h$, the function $\mathcal X_2$ satisfies the estimate
\begin{equation*}
  \|e^{-\frac{2i\theta}h}\mathcal X_2
\|_{\gamma,a,2(1-a)n_u(t)-n_u(M_{12})-ch, 2(1-a)n_d(t)-n_d(M_{12})-ch}\le Ch.
\end{equation*}
Therefore, as $2an_{s}(t)<1$, $s\in\{d,u\}$, for sufficiently small $h$, for all $z\in \Pi_{\gamma,a}$ 
\begin{equation}\label{est:X2}
  |e^{-\frac{2i\theta(z)}h}\mathcal X_2(z)|\le Ch e^{-2\pi (n_s(t)-n_s(M_{12})) |y|}.
\end{equation}
Furthermore, as $n_s(t)\ge n_s(M_{12})$,  
for all $z\in \Pi_{\gamma,a}$\;, we get
\begin{equation}\label{est:X2:simple}
 |e^{-\frac{2i\theta(z)}h}\mathcal X_2(z)|\le Ch.
\end{equation}
{\bf 4.} \ It follows from~\eqref{eq:X2} and~\eqref{eq:X10} that $\mathcal X$, 
the vector with the elements $\mathcal X_1$ and $\mathcal X_2$, satisfies 
equation~\eqref{eq:3} in $\Pi_{\gamma,a}$. 
\\
{\bf 5.} \ Let $A$ be an admissible subdomain of $K$ containing $\gamma$
with some its $\delta$-neighbor\-hood.
Let us assume that $h$ is sufficiently small, and prove that
the function $\mathcal X$ is analytic in $A$ and satisfies \cref{eq:withS} 
if $z,z+h\in A$.
\\
The function $\mathcal X$ is analytic between  
$\gamma-ah$ and $\gamma+(a+1)h$. Indeed, $\mathcal X$ is defined and  
analytic between  $\gamma-ah$ and $\gamma+ah$.  This and  the definition 
of $\mathcal L_+$ imply that  the function $\mathcal L_+ (S \mathcal X)$ is 
analytic between $\gamma-ah$ and $\gamma+(a+1)h$. As 
$\mathcal X$ satisfies~\cref{eq:3}, this implies that it is also analytic there. 
\\
Let us assume that $z+h,z$ are located between $\gamma-ah$ and 
$\gamma+(a+1)h$. Computing the difference $\mathcal{L}_+(S\mathcal X)(z+h)-
\mathcal{L}_+(S\mathcal X)(z)$ by means of the residue theorem, 
we check that $\mathcal X$ satisfies \cref{eq:withS}, and so  
$\mathcal X(z+h)=(I+S(z))\mathcal X(z)$.
\\
This allows to continue $\mathcal X$ analytically  from the strip bounded 
by $\gamma-ah$ and $\gamma+(a+1)h$ into the part of $K$ located on the 
right of $\gamma$. 
\\
In view of~\eqref{eq:S} and~\eqref{est:s}, we have 
$\det(1+S(z))=1+O(h^2e^{-2\pi(1-2ch)|y|})$ in $A$.
So, for sufficiently small $h$, for $z\in A$ the matrix $I+S(z)$ is invertible,  and
$\mathcal X(z)=(I+S(z))^{-1}\mathcal X(z+h)$. This  allows to 
continue $\mathcal X$ analytically from the strip bounded by $\gamma-ah$ 
and $\gamma+(a+1)h$ in the part of $A$ located on the left of $\gamma$.
\\
By construction $\mathcal X$ satisfies \cref{eq:withS} if $z,z+h\in A$.
\\
{\bf 6.} Let us show that, for sufficiently small $h$,  estimates \eqref{est:X1} 
and \eqref{est:X2:simple} are valid and uniform in any compact subset of $A$.
\\
Let $z^0\in A$, and let $\gamma^0$ be a strictly canonical curve 
containing $z^0$ and contained in $K$ with some its $\delta$-neighborhood.
Clearly, a strictly canonical curve remains strictly canonical if we 
deform it only in a $\delta$-neighborhood of its point and if this deformation 
is sufficiently small in $C^1$-topology.   Therefore,  $z^0$ is an internal point of a simply 
connected domain $D^0\subset A$ bounded by two strictly canonical curves  
$\gamma^1\subset K$ and $\gamma^2\subset K$ that coincide with 
$\gamma^0$  outside a neighborhood of the point $z^0$.
\\
Let $j\in\{1,2\}$. Deforming in~\cref{eq:X10}  the integration path $\gamma$ 
inside $K$ to $\gamma^j$, one shows that $\mathcal X_1$ is a solution to this 
equation in $H_{\gamma^j,a,0,0}$, and, therefore, satisfies estimate 
\eqref{est:X1} uniformly in $z\in \Pi_{\gamma^j,a}$. Deforming the integration path
in~\eqref{eq:X2} to $\gamma^j$, we come to  estimate \eqref{est:X2:simple}
uniform in $z\in \Pi_{\gamma^j,a}$.
\\
As estimates \eqref{est:X1} and \eqref{est:X2:simple} hold on the boundary of $D^0$,
by the Maximum modulus principle for analytic functions, they
hold  in $D^0$. This implies the needed.
\\
{\bf 7.} \  For sufficiently large  $Y>0$, and for sufficiently small $h$, the 
vector $\mathcal X$ satisfies estimates  \eqref{est:X1} and 
\eqref{est:X2}  in the domain $A(Y)=\{z\in A\,:\,|\im z|>Y\}$. 
\\
Indeed, when proving Lemma 8.1 from \cite{F-Shch:18} (steps 2 and 3 of the 
proof), we checked that if $Y$ is sufficiently large, and $h$ is sufficiently small, 
then there is a constant $C_0>0$ independent of $h$ and  such that  for any point 
$z_0\in A(Y)$, there is a canonical curve $\gamma\subset A$ containing $z_0$ and 
such that the estimates of \Cref{pro:LK} with $C=C_0$  hold. This implies
the needed.   
\\
{\bf 8.} \ As $\mathcal X$ satisfies~\eqref{eq:withS}, one constructs a vector 
solution  to~\eqref{main} by the formulas
\begin{gather}\nonumber
 \hspace{-6cm} \Psi^+(z)=\Psi_0(z)\; 
\begin{pmatrix} e^{\phi_1(z)} & 0 \\ 0 & e^{\phi_2(z)}\end{pmatrix}\, \mathcal X(z)\\
\label{eq:PsiX}
\hspace{1cm} =e^{\frac {i\theta(z)}h+\phi_1(z)}\mathcal X_1(z)
\left( V^+(z)+e^{-\frac {2i\theta(z)}h+\phi_2(z)-\phi_1(z)}\ 
\frac{\mathcal X_2(z)}{\mathcal X_1(z)}\,V^-(z)\right).
\end{gather}
Being defined and analytic in a strip between $\gamma-ah$ and $\gamma+ah+h$,   
the solution $\Psi^+$ can be analytically continued up to an entire function just by
means of \cref{main}.
\\
{\bf 9.} \ Using~\eqref{eq:PsiX}, 
\eqref{est:X1}, \eqref{est:X2:simple} and \eqref{est:phis}, we get 
\begin{equation}\label{eq:PsiX:2}
\Psi^+(z)=e^{\frac {i\theta(z)}h+\phi_1(z)}\mathcal X_1(z)\left(V^+(z)+
O(h e^{C\,h |y|})V^-(z)\right).
\end{equation}
In view of the definition of $\theta$, see \eqref{T:as}, \ \eqref{eq:PsiX:2} implies 
representation~\eqref{as:main} locally uniform in $z$. This completes the proof 
of \Cref{th:main:local} for the solution $\Psi^+$ in the case 
we study (when $e^{i p(z)}\to 0$ as $|\im z|\to\infty$.)
\\
{\bf 10.} \ Now, let us prove \Cref{th:main:global} for the solution $\Psi^+$.
\\
We  fix $Y>0$ sufficiently large. 
In view of~\eqref{d:ev}, \eqref{d:r-pm}, \eqref{as:omegas} and~\eqref{as:omegas:d},
for $s\in\{u,d\}$ and $z\in \C_s(Y)$, we get the estimates:
\begin{equation*}
\begin{split}
V_1^+(z)&\asymp e^{-\pi \,(n_s(t) -2 n_s(M_{12}))\,|y|},\qquad
V_2^+(z)\asymp e^{\pi n_s(t)\,|y|}, \\
V_1^-&(z)=O( e^{-\pi n_s(t)\, |y|}),\qquad
V_2^-(z)=O( e^{-\pi n_s(t)\,|y|}).
\end{split}
\end{equation*}
This,~\eqref{est:X1} and \eqref{est:X2}
imply that, for $j=1,2$,  for sufficiently small $h$,
\begin{equation}\label{est:V}
e^{-\frac {2i\theta(z)}h+\phi_2(z)-\phi_1(z)}\ 
\frac{\mathcal X_2(z)}{\mathcal X_1(z)}\,\frac{V_j^-(z)}{V_j^+(z)}=
O(he^{-2\pi( n_s(t)-2c\, h)\,|y|})=O(h)
\end{equation}
uniformly in $z\in A(Y)$. This and \eqref{eq:PsiX} imply 
the statement of \Cref{th:main:global} on $\Psi^+$ in the case 
we study.
\\
To construct  the solution $\Psi^-$ , one proceeds as
suggested in section 9 of \cite{F-Shch:18}.
\\
We have studied the case of~\eqref{hyp:exp}. The complementary cases 
are analyzed similarly; in section 10 of \cite{F-Shch:18}, 
we indicated the way  to do this. We omit further details.
%
%

%
%

\begin {thebibliography}{99}
%
\bibitem{A-J:09}
{\sc A. Avila, S. Jitomirskaya},  
{\em The Ten Martini Problem},
Annals Math., 170 (2009), pp.~303--342.
%
\bibitem{B:84}
{\sc M. Berry},
{\em Quantal phase factors accompanying adiabatic changes},
Proc. Roy. Soc. London, A,  392(1984), pp.\; 45--57.
%
\bibitem{B-F:94}
{\sc V.~Buslaev and  A.~Fedotov}, 
{\em The complex WKB method for Harper equation},
St. Petersburg Math. J., 6(1995), pp.\; 495--517.
%
\bibitem{B-F:01}
{\sc V.~Buslaev and A.~Fedotov}, 
{\em On difference equations with periodic 
coefficients}, Adv. Theor. Math. Phys., 6 (2001), pp.~1105--1168.
%
\bibitem{B-F:94a}
{\sc V.~Buslaev and A.~Fedotov}, 
{\em The monodromization and Harper equation}, 
S\'eminaire sur les \'Equations aux D\'eriv\'ees Partielles 1993-1994, Exp. no EXXI,  
\'Ecole Polytech., Palaiseau, 1994, 23 pp.
%
\bibitem{Dobro}
{\sc S.~Y. Dobrokhotov and A.~V. Tsvetkova}, 
{\em On lagrangian manifolds  related to asymptotics of {H}ermite polynomials}, 
Math. notes, 110 (2018).
%
\bibitem{Eck}
{\sc A.~Eckstein}, {\em Unitary  reduction  for  the  two-dimensional  
Schr{\"o}dinger  operator  with strong   magnetic   field},
Math. Nachr., 282(2009), pp.~504-525.
%
\bibitem{Fe:93}
{\sc M.V.~Fedoryuk}, 
{\em Asymptotic Analysis. Linear Ordinary Differential Equations},
Berlin-Heidelberg GmbH, Springer-Verlag, 2009.
%
\bibitem{F:13} 
{\sc A.~Fedotov}, 
{\em Monodromization method in the theory of almost-periodic equations},  
St. Petersburg Math. J.,  25 (2014), pp.~303-325.

\bibitem{F-K:02}
{\sc A.~Fedotov and F.~Klopp}, {\em Anderson transitions for a family 
of almost periodic {S}chr{\"o}dinger equations in the adiabatic case}, 
Commun. in Math.  Phys., 227 (2002), pp.~1--92.

\bibitem{F-K:04} 
{\sc A.~Fedotov and F.~Klopp},  
{\em Geometric tools of the adiabatic complex WKB method}, 
Asymptotic analysis, 39(2004), pp.~309-357.

\bibitem{F-K:05a}
{\sc A.~Fedotov and F.~Klopp}, {\em Strong resonant tunneling, level repulsion
and spectral type for one-dimensional adiabatic quasi-periodic 
Schr\"odinger operators}, Annales Scientifiques de l'Ecole Normale Superieure, 
4e serie, 2005, pages 889 - 950.

\bibitem{F-K:18} 
{\sc A.~Fedotov and F.~Klopp},  
{\em The complex WKB method for difference equations and 
Airy functions}, to be published in SIAM J. Math.An.;
https://hal.archives-ouvertes.fr/hal-01892639.
%
\bibitem{F-K:18a} 
{\sc A.~Fedotov and F.~Klopp},  
{\em WKB asymptotics of meromorphic solutions to
difference equations}, Applicable Analysis, 
2019, DOI: 10.1080/00036811.2019.1652735.
%
%
\bibitem{F-Shch:17}
{\sc A.~Fedotov and  E.~Shchetka}, {\em  The complex WKB method for 
difference equations in bounded domains}, J. Math. Sci. (New York), 224 (2017), 
pp.\;157--169.
%
\bibitem{F-Shch:17DD}
{\sc A.~Fedotov, E.~Shchetka}, {\em Berry phase for difference equations}.
In {\it Days on Diffraction, 2017}. Institute of Electrical and Electronics Engineers 
Inc., 2017, 113-115.
%
\bibitem{F-Shch:18}
{\sc A.~Fedotov, E.~ Shchetka}, 
{\em Complex WKB method for a difference Schr\"odinger equation with the potential 
being a trigonometric polynomial}, St. Petersburg Math. J., 29(2018), 363–381.
%
\bibitem{G-at-al}
{\sc J.~S. Geronimo, O.~Bruno, and W.~V. Assche}, {\em {WKB} and turning point
  theory for second-order difference equations}, Operator theory: advances and
  app., 69 (1992), pp.~269--301.
%
\bibitem{GHT:89}
{\sc J.~P. Guillement, B.~Helffer, and P.~Treton}, {\em Walk inside
  {H}ofstadter's butterfly}, J. Phys. France, 50 (1989), pp.~2019--2058.
%
\bibitem{H-S:88} 
{\sc B.~Helffer, and J.~Sj{\"o}strand},  {\em Analyse semi-classique pour 
l'{\'e}quation de Harper (avec application {\`a} l'{\'e}quationde Schr{\"o}dinger 
avec champ  magn{\'e}tique)}, M{\'e}moires de la SMF (nouvelle s{\'e}rie), 
34(1988), pp.~1-113.
%
\bibitem{L-Z:03}
{\sc M.~A. Lyalinov, and N.~Y. Zhu}, {\em
A solution procedure for second-order difference equations and its
application to electromagnetic-wave diffraction in a wedge-shaped region},
Proc. R. Soc. Lond. A, 459(2003), pp.~3159-3180.
%
\bibitem{S:75}
{\sc Y. Sibuya}, {\em Global theory of a second order linear ordinary
differential equation with a polynomial coefficient}, North Holland/American
Elsevier,  1975.
%
\bibitem{S:83}
{\sc B. Simon},
{\em Holonomy, the quantum adiabatic theorem and Berry's phase},
\emph{Phys. Rev. Lett.}, 51(1983), pp.~2167--2170.
%
\bibitem{springer}
{\sc G. Springer}, 
{\em Introduction to Riemann surfaces}, New York, Addison–Wesley, 1957.
%
\bibitem{W:87}
{\sc W.~Wasow}, {\em Asymptotic expansions for ordinary differential
equations}, Dover Publications, New York, 1987.
%
\bibitem{Wi}
{\sc M.~Wilkinson}, {\em  An exact renormalization group for Bloch electrons in a
magnetic field}, J. Phys. A: Math. Gen., 20 (1987), pp.~4337--4354.
%
\end{thebibliography}
\end{document}